\documentclass[twocolumn, 10pt, final]{IEEEtran}

\usepackage{cite}

\ifCLASSINFOpdf
  \usepackage[pdftex]{graphicx}
  \DeclareGraphicsExtensions{.eps}
\else
\fi

\usepackage{stfloats}
\usepackage[dvipsnames]{xcolor}
\usepackage[
  linktocpage,
  pagebackref=false,
  hypertexnames=false
]{hyperref}

\usepackage{prettyref}      
\usepackage{placeins}
\usepackage[scr=boondoxo]{mathalfa}
\usepackage{mathtools}
\usepackage{bbm}
\usepackage{bm}
\usepackage{amsthm}
\usepackage{amsfonts}
\usepackage{amssymb}
\usepackage{amsbsy}
\usepackage{amsmath}
\usepackage{xfrac}
\usepackage{graphicx}
\usepackage{epstopdf}
\usepackage{tabularx,colortbl}
\usepackage{breqn}
\usepackage{framed}
\usepackage{algorithm}
\usepackage{algpseudocode}
\usepackage{siunitx}

\newrefformat{fig}{Fig. \ref{#1}}
\newrefformat{tab}{Table \ref{#1}}
\newrefformat{eq}{{\normalfont (\ref{#1})}}
\newrefformat{thm}{Theorem \ref{#1}}
\newrefformat{cor}{Corollary \ref{#1}}
\newrefformat{def}{Definition \ref{#1}}
\newrefformat{lem}{Lemma \ref{#1}}
\newrefformat{prop}{Proposition \ref{#1}}
\newrefformat{assum}{Assumption \ref{#1}}
\newrefformat{cond}{\conditionname~\ref{#1}}

\providecommand{\conditionname}{Condition}
\providecommand{\corollaryname}{Corollary}
\providecommand{\definitionname}{Definition}
\providecommand{\asuumptionname}{Assumption}
\providecommand{\lemmaname}{Lemma}
\providecommand{\propositionname}{Proposition}
\providecommand{\theoremname}{Theorem}

\providecommand{\conditionname}{Condition}

\makeatletter 
\let\oldforeign@language\foreign@language
\DeclareRobustCommand{\foreign@language}[1]{%
   \lowercase{\oldforeign@language{#1}}}
\theoremstyle{plain}
\newtheorem{thm}{\protect\theoremname}
\theoremstyle{definition}

\theoremstyle{definition}
\newtheorem{assum}[thm]{\protect\asuumptionname}
\theoremstyle{plain}
\newtheorem{cor}[thm]{\protect\corollaryname}
\theoremstyle{plain}
\newtheorem{prop}[thm]{\protect\propositionname}
\theoremstyle{plain}
\newtheorem{lem}[thm]{\protect\lemmaname}

\newcounter{assum}[section]
\renewenvironment{assum}[1][]{\refstepcounter{assum}\par \medskip
   \noindent \textbf{Assumption~\theassum. #1} }{\par \medskip}

\theoremstyle{remark}
\newtheorem{remark}{\bf Remark}
\newtheorem{cond}{\bf Condition}
\newtheorem{definition}{\bf Definition}

\allowdisplaybreaks  

\newcounter{MYtempeqncnt}

\newenvironment{floateq}[5][h!]
  {\begin{figure*}[#1]
  \normalsize
  \setcounter{MYtempeqncnt}{\value{equation}}
  \setcounter{equation}{#2}
  \ifcase #3 \hrulefill \else \fi
  \begin{align}
  #5
  \end{align}
  \setcounter{equation}{\value{MYtempeqncnt}}
  \ifcase #4 \hrulefill \else \fi
  \vspace*{0pt}
  \end{figure*}
  }

\def\cleartheorem#1{%
    \expandafter\let\csname#1\endcsname\relax
    \expandafter\let\csname c@#1\endcsname\relax
}

\begin{document}
%
\title{Non-Asymptotic Guarantees for {Reliable}\\Identification of Granger Causality via the LASSO}
%
%
%

\author{Proloy~Das,~\IEEEmembership{Student Member,~IEEE,}
        and~Behtash~Babadi,~\IEEEmembership{Member,~IEEE}
\thanks{P. Das is with the Department of Anesthesia, Critical Care and Pain Medicine, Massachusetts General Hospital, Boston,
MA, 02114 USA (e-mail: pdas6@mgh.harvard.edu).}
\thanks{B. Babadi is with the Department
of Electrical and Computer Engineering and the Institute for Systems Research, University of Maryland, College Park,
MD, 20742 USA (e-mail: behtash@umd.edu).}
\thanks{This work was supported in part by the National Science Foundation Awards No. CCF1552946 and ECCS1807216 and the National Institutes of Health Award No. 1U19NS107464 (Corresponding author: Proloy Das).}}

\maketitle

\begin{abstract}
Granger causality is among the widely used data-driven approaches for causal analysis of time series data with applications in various areas including economics, molecular biology, and neuroscience. Two of the main challenges of this methodology are: 1) over-fitting as a result of limited data duration, and 2) correlated process noise as a confounding factor, both leading to errors in identifying the causal influences. Sparse estimation via the LASSO has successfully addressed these challenges for parameter estimation. However, the classical statistical tests for Granger causality resort to asymptotic analysis of ordinary least squares, which require long data {duration} to be useful and are not immune to confounding effects. In this work, we {address this disconnect} by introducing a LASSO-based statistic and studying its non-asymptotic properties under the assumption that the true models admit sparse autoregressive representations. We establish {fundamental limits for reliable} identification of Granger causal influences using the proposed LASSO-based statistic. We further characterize the false positive error probability and test power of a simple thresholding rule for identifying Granger causal effects and provide two methods to set the threshold in a data-driven fashion. We present simulation studies and application to real data to compare the performance of our proposed method to ordinary least squares and {existing LASSO-based methods} in detecting Granger causal influences, which corroborate our theoretical results.
\end{abstract}


\ifCLASSOPTIONpeerreview
\begin{center} \bfseries EDICS Category: 3-BBND \end{center}
\fi
%
\IEEEpeerreviewmaketitle

\section{Introduction}
\IEEEPARstart{R}{eliable} identification of causal influences is one of the central challenges in time series analysis, with implications for various domains such as economics \cite{greene2003econ}, neuroscience \cite{BRESSLER2011323,FRISTON2013172,Seth3293} and computational biology \cite{klipp2005systems,feng2007networks}.
Granger causal (GC) characterization of time series is among the widely used methods in this regard. This framework was pioneered by Granger \cite{granger1969}, with subsequent key generalizations provided by Geweke \cite{geweke1982measurement,geweke1984lineardep}.
The notion of GC influence pertains to assessing the improvements in predicting the future samples of one time series by incorporating the past samples of another one. 

{While causality, as the relationship between cause and effect, is a philosophically well-defined concept \cite{psillos2014causation}, it eludes a universal definition in empirical sciences and engineering.} Granger causality is one of many definitions used in time series models (see \cite{marinescu2018quasi,pearl2009causality} for other notions), with an explicit data-driven form that admits statistical testing. The stochastic nature of the time series model, i.e., the uncertainty is the central feature of GC definition. {Unlike regression analysis that merely reflects correlational association, Granger's notion of causality probes if two time series are \emph{temporally related} \cite{grangerForecastingEconomicTime2014}, in terms of the precedence established by the direction of time flow, i.e., if one time series precedes, and thus forecasts, another.}
In principle, given two time series $x_{t}$ and $y_{t}$, one asserts that $y_t$ has a GC influence on $x_t$ when the posterior {conditional} densities $p(x_{t}|x_{t-1},x_{t-2},\cdots,y_{t-1},y_{t-2},\cdots\vphantom{x_{t}})$ and $p(x_{t}|x_{t-1},x_{t-2},\cdots\vphantom{x_{t}})$ differ significantly.
However, estimating these posterior densities from the observed data is a difficult task in general, and requires additional modeling assumptions.
A popular set of such assumptions pertains to parametric multivariate autoregressive (MVAR) models along with {certain} distributional specifications (e.g., zero-mean Gaussian process noise).
In these models, the aforementioned posterior densities can be fully characterized by the estimates of parameters and prediction error variances. As such, the classical notion of Granger causality focuses on the prediction error variances: one first aims at predicting $x_{t+1}$ by a linear combination of the joint past observations $\{x_{t}, \cdots x_{0}\}$, $\{y_{t}, \cdots y_{0}\}$ (i.e., the \emph{full} model), followed by {repeating} this task by excluding the past observations of ${y}_t$ (i.e., the \emph{reduced} model). If the prediction error variance in the former case is significantly smaller than the latter, we say that $y_t$ has a GC influence on $x_t$ {(See \prettyref{sec:classical GC} for a formal definition of GC)}. 

Conventionally, {the optimal linear predictors in the mean square error sense,} are obtained by the ordinary least squares (OLS), and the model orders are determined by the AIC \cite{Akaike1974model} or BIC \cite{schwarz1978} procedures.
Then, the GC measure is defined as the logarithmic ratio of the two prediction error variances, and its statistical significance is assessed based on the corresponding asymptotic distributions \cite{kim2011granger,Wald1943,Davidson1970}.
While the aforementioned procedure is relatively simple to carry out, it faces two key challenges.
First, in order to obtain reliable MVAR parameter estimates via OLS, a relatively long observation horizon is required. In datasets with small sample size (e.g., gene expression data \cite{zou2009granger}), the regression models typically over-fit the observed data, causing both parameter estimation and model order selection to break down \cite{Seth3293,seth2010matlab}.
In addition, AIC/BIC may restrict the order of the MVAR in a way that the resulting model fails to capture the complex and long-range dynamics of the underlying couplings \cite{BRESSLER2011323,MITRA1999691}. Secondly, correlated process noise arising from latent processes, may lead to misidentification of GC influences, which is often referred to as the confounding effect {\cite{bahadori2013examination}}. 

These challenges have been successfully addressed in the context of regularized MVAR estimation {\cite{Goldenshluger2001,wang2007regression,nardi2011autoregressive,han2013transition,kazemipour2017sampling,basu2015regularized,wong2016lasso,SKRIPNIKOV2019164,basu2019low}.} 
In particular, the theory of sparse estimation via the LASSO \cite{Tibshirani1996Lasso,loh2011high,buhlmann2011statistics,negahban2012,hastie2015statistical} provides a principled methodology for simultaneous parameter estimation and model selection in high dimensional MVAR models \cite{wang2007regression,HSU20083645,REN20101705,nardi2011autoregressive,ren2013model,han2013transition,basu2015regularized,wong2016lasso,kazemipour2017sampling}.
In addition, the Oracle property of the LASSO in presence of correlated noise ensures robust recovery of the set of MVAR parameters arising from the direct GC influences while discarding any spurious couplings due to correlated process noise, thus {alleviating} the confounding effect.
The LASSO and its variants have already been utilized in existing work to identify graphical GC influences based on the recovered sparsity patterns \cite{arnold2007temporal,fujitaModelingGeneExpression2007,Lozano2009GGGM,Shojaie2010GGC,liu2012sparse,Basu2015NGC,MICHAILIDIS2013326,shimamuraRecursiveRegularizationInferring2009}. These methods construct the GC graph based on the estimated model parameters, either directly \cite{arnold2007temporal} or by appropriate thresholding \cite{Basu2015NGC,MICHAILIDIS2013326} to control false positive errors.
This idea has even been extended to time series models that account for nonlinear dynamics using structured multilayer perceptrons or recurrent neural networks \cite{tank2021neural}.
Another related class of results uses de-biasing techniques in order to construct confidence intervals and thereby identify the significant GC interactions \cite{tang2012measuring,van2014asymptotically,van2017efficiency,javanmard2014confidence,javanmard2018debiasing,javanmard2019false,zhangConfidenceIntervalsLow2014}.
There is, however, an evident disconnect between these LASSO-based approaches and the classical OLS-based GC inference: while the LASSO-based approaches aim at identifying the GC effects based on consistent {\emph{parameter estimates}} in the non-asymptotic regime, the classical GC methodology relies on the comparison of the {\emph{prediction errors}} across the \emph{full} and \emph{reduced} models by resorting to asymptotic distributions.

In this paper, we {address the aforementioned disconnect} between the currently available LASSO-based approaches and the classical OLS-based GC inference by unifying these two approaches via introducing a new LASSO-based GC statistic that resembles the classical GC measure, and by leveraging the consistency properties of the LASSO to characterize the non-asymptotic properties of said GC statistic.
In particular, we consider a canonical bivariate autoregressive (BVAR) process with correlated process noise.
We then propose a likelihood-based {scaled} F-statistic as the relevant GC statistic, which we call the LGC statistic, and study its non-asymptotic properties under both the presence and absence of a GC influence.
Our analysis reveals that the well-known sufficient conditions of the LASSO for stable BVAR estimation are also sufficient for accurate detection of the GC influences, if the strength of the GC effect satisfies additional conditions. We also show that the conditions on the strength of the GC effect pose a fundamental limit for GC identification using LGC. Furthermore, we characterize the false positive error probability and test power of a simple thresholding scheme for identification of GC influences and provide two methods to set the threshold in a data-driven fashion.

{We present extensive simulation studies to compare the performance of the proposed LGC-based method with the classical OLS-based and two of the existing LASSO-based approaches in detecting GC influences. These studies demonstrate the validity of our theoretical claims, explore the key underlying trade-offs, and evaluate our contributions in the context of existing work.}
We also present an application to experimentally-recorded neural data from general anesthesia to assess the GC influence of the local field potential (LFP) on spiking activity.
Our results based on LGC analysis corroborate existing hypotheses on the GC influence of LFP in mediating local spiking activity, whereas these effects are concealed by the classical GC analysis due to significant over-fitting.

In summary, our main contribution is to extend the theoretical results of the LASSO to the classical characterization of GC influences, to identify the key trade-offs in terms of sampling requirements and strength of the GC effects that result in {reliable} GC identification, to characterize the false positive error probability and test power of GC detection by thresholding LGC, and to provide data-driven methods to set the test threshold in practice.

The rest of this paper is organized as follows: \prettyref{sec:Previous-Works} provides background and our problem formulation.
Our main theoretical and methodological contributions are given in \prettyref{sec:Main-contribution}.
\prettyref{sec:experimental validation} presents application to simulated and experimentally-recorded neural data, followed by our concluding remarks in \prettyref{sec:Conclusions}.
\section{Background and Problem Formulation\label{sec:Previous-Works}}
Throughout this work, we use regular and boldface lowercase letters to denote scalars and vectors, respectively. Matrices are denoted by boldface uppercase letters. The transpose and Hermitian of a matrix $\mathbf{M}$ are denoted by $\mathbf{M}^\top$ and $\mathbf{M}^H$, respectively. The $\ell_\alpha$-norm of a vector $\mathbf{v} \in \mathbbm{R}^n$ is denoted by $\left \Vert \mathbf{v} \right \Vert_\alpha = \left(\sum_{i=1}^n \vert v_i \vert^\alpha\right)^{1/\alpha}$ for $\alpha \ge 1$. In addition, we use the notation $\left \Vert \mathbf{v} \right \Vert_0 = \sum_{i=1}^n \mathbbm{1}\left[v_i \neq 0\right]$ to denote the number of non-zero elements of the vector $\mathbf{v}$, {and drop} the subscript of the $\ell_2$-norm of $\mathbf{v}$ and denote it by $\left\Vert \mathbf{v} \right\Vert$. For a matrix $\mathbf{M}$, the matrix norm induced by the $\ell_\alpha$-norm is defined as $\Vert\mathbf{M}\Vert_\alpha := \sup\limits_{\mathbf{x} \neq \mathbf{0}} \, \Vert\mathbf{M}\mathbf{x}\Vert_\alpha / \Vert\mathbf{x}\Vert_\alpha$. We denote the minimum and maximum eigenvalues of a $\mathbf{M}$ by $\Lambda_{\min}(\mathbf{M})$ and $\Lambda_{\max}(\mathbf{M})$, respectively.
\subsection{Granger Causality in a Canonical BVAR Regression Model\label{sec:classical GC}}
Consider finite-duration observations from two time series $x_t$ and $y_t$, given by $\left\{ x_{t},y_{t}\right\} _{t=-p+1}^{n}$, where $n$ is the sample size and $p$ is the order.
The BVAR$(p)$ model can be expressed as:
\begin{align}
\!\!\!\begin{bmatrix}x_{t}\\
y_{t}
\end{bmatrix}\!=\!\mathbf{A}_{1}\begin{bmatrix}x_{t-1}\\
y_{t-1}
\end{bmatrix}\!+\!\mathbf{A}_{2}\begin{bmatrix}x_{t-2}\\
y_{t-2}
\end{bmatrix}\!+\!\cdots\!+\!\mathbf{A}_{p}\begin{bmatrix}x_{t-p}\\
y_{t-p}
\end{bmatrix}\!+\!\begin{bmatrix}\epsilon_{t}\\
\epsilon'_{t}
\end{bmatrix},
\label{eq:true-MVAR(p)}
\end{align}
with $\mathbf{A}_{i}\in\mathbbm{R}^{2\times2}\text{ for }i\in\{1,2,\cdots,p\}$ denoting the BVAR parameters and $[\epsilon_t, \epsilon'_t]^\top$ denoting the process noise with known distribution.
It is commonly assumed that $[\epsilon_t, \epsilon'_t]^\top \sim \mathcal{N}(\mathbf{0}, \boldsymbol{\Sigma}_\epsilon)$.
Using this BVAR($p$) model and considering $\left\{ x_{t},y_{t}\right\} _{t=-p+1}^{0}$ as the initial condition, one can form a prediction model of $x_t$ as follows:
\begin{align}
\mathbf{x}=\mathbf{X}\boldsymbol{\theta}+\boldsymbol{\epsilon},
\label{eq:MVAR(p)}
\end{align}
where the response $\mathbf{x}$, regressors $\mathbf{X}$, and residuals $\boldsymbol{\epsilon}$ are defined as in \prettyref{eq:regressors}.

\begin{floateq}[b!]{2}{0}{1}
{\mathbf{x}=\begin{bmatrix}x_{n}\\
x_{n-1}\\
\vdots\\
x_{1}
\end{bmatrix}, \hspace*{0.5em}
\mathbf{X}=\begin{bmatrix}x_{n-1} & \cdots & x_{n-p} & y_{n-1} & \cdots & y_{n-p}\\
x_{n-2} & \cdots & x_{n-p-1} & y_{n-2} & \cdots & y_{n-p-1}\\
\vdots & \cdots & \vdots & \vdots & \cdots & \vdots\\
x_{0} & \cdots & x_{-p+1} & y_{0} & \cdots & y_{-p+1}
\end{bmatrix},\hspace*{0.5em}
\boldsymbol{\epsilon}=\begin{bmatrix}\epsilon_{n}\\
\epsilon_{n-1}\\
\vdots\\
\epsilon_{1}
\end{bmatrix}.\label{eq:regressors}}
\end{floateq}
The regression coefficients $\boldsymbol{\theta}$ consist of $2p$ parameters: $\{\theta_{i}\}_{i=1}^{p}$, representing the autoregression coefficients obtained from $(\mathbf{A}_i)_{1,1}$, $i=1,2,\cdots, p$ and $\{\theta_{i}\}_{i=p+1}^{2p}$ representing the cross-regression coefficients obtained from $(\mathbf{A}_i)_{1,2}$, $i=1,2,\cdots, p$.
Hereafter, we denote the true coefficients by $\boldsymbol{\theta}^* \in \mathbb{R}^{2p}$. Also, for a generic coefficient vector $\boldsymbol{\theta} \in \mathbb{R}^{2p}$, the corresponding auto- and cross-regression components are denoted by $\boldsymbol{\theta}_{(1)} \in \mathbb{R}^p$ and $\boldsymbol{\theta}_{(2)} \in \mathbb{R}^p$, respectively, i.e., $\boldsymbol{\theta} =: [\boldsymbol{\theta}_{(1)}; \boldsymbol{\theta}_{(2)}]$.

In the context of this bivariate model, let $\mathscr{U}_t$ be the sigma-field generated by all random variables up to and including time $t-1$, and $\left({\mathscr{U}\text{\textbackslash} y}\right)_t$ be the sigma-field generated by all random variables up to and including time $t-1$, except $\{y_{t-1}, y_{t-2}, \cdots\}$. {In addition, let the variance of a random variable, $Z$ conditional on these sigma-fields be $\sigma^2(Z|\mathscr{U}_t)$ and $\sigma^2(Z|(\mathscr{U}\backslash y)_t)$. With this notation,  Granger Causality (GC) is formally defined for Gaussian BVAR models as follows}: 

\begin{definition}[Granger Causality \cite{granger1969}]
\label{def:gc}
We say $y_t$ has a Granger causal link to $x_t$, if $\sigma^2(x_t|\mathscr{U}_t) < \sigma^2(x_t|(\mathscr{U}\backslash y)_t)$, i.e., we are better able to predict $x_t$ using all available data than excluding $y_t$.
\end{definition}
\noindent While \prettyref{def:gc} applies to any time $t$, in practice the GC is assessed within blocks of the time-series data, e.g., for $0\le t \le n$, since estimating the prediction variance reliably for all time points and using a single trial of the time-series data is challenging.
A convenient framework to test the GC influence of $y_t$ on $x_t$ is to pose it as hypothesis testing, with the null hypothesis $H_{y\mapsto x,0}:\boldsymbol{\theta}^*_{(2)}=\boldsymbol{0}$. To this end, one considers the following BVAR\((p)\) models:
\setcounter{equation}{\numexpr\value{equation}+1}
\begin{subequations}
\begin{align}
\text{Full Model: }&\mathbf{x}= \mathbf{X}\boldsymbol{\theta}+\boldsymbol{\epsilon}, \\
\text{Reduced Model: }&\mathbf{x}=  \mathbf{X}\widetilde{\boldsymbol{\theta}}+\widetilde{\boldsymbol{\epsilon}}, \, \, \text{with } \widetilde{\boldsymbol{\theta}}_{(2)}=\boldsymbol{0}.\label{eq:mvar}
\end{align}
\end{subequations}
In words, in the \emph{full} model, all columns of $\mathbf{X}$ are used to estimate $\mathbf{x}$, but in the \emph{reduced} model, only the first $p$ columns are used.
The conventional GC measure \cite{geweke1984lineardep,granger1969,geweke1982measurement} is then defined as the logarithmic ratio of the residual variances: $\mathcal{F}_{y\mapsto x} :=\log\left(\text{var}(\widetilde{\boldsymbol{\epsilon}})/\text{var}({\boldsymbol{\epsilon}})\right)$.
Note that when the residuals are Gaussian, $\mathcal{F}_{y\mapsto x}$ is the log-likelihood ratio statistic.
Given that the \emph{reduced} model is nested within the \emph{full} model, we have $\mathcal{F}_{y\mapsto x}\geq0$. 

In order to compute $\mathcal{F}_{y\mapsto x}$ from the time series data, empirical residual variances are used based on OLS parameter estimates under both models \cite{greene2003econ}:
\begin{subequations}
\begin{align}
\widehat{\boldsymbol{\theta}}_{\text{OLS}}&=\underset{\boldsymbol{\theta}}{\text{argmin}}\hspace*{2mm}\frac{1}{n}\left\Vert \mathbf{x}-\mathbf{X}\boldsymbol{\theta}\right\Vert ^{2}, \\
\widehat{\widetilde{\boldsymbol{\theta}}}_{\text{OLS}}&=\underset{\boldsymbol{\theta}: \boldsymbol{\theta}_{(2)}=\boldsymbol{0}}{\text{argmin}}\hspace*{2mm}\frac{1}{n}\left\Vert \mathbf{x}-\mathbf{X}\boldsymbol{\theta}\right\Vert ^{2}.
\label{eq:ols}
\end{align}
\end{subequations}
The estimated ${\mathcal{F}}_{y\mapsto x}$ is a random variable over \(\mathbbm{R}_{\geq0}\), and typically has a non-degenerate distribution.
Thus, a non-zero ${\mathcal{F}}_{y\mapsto x}$ does not necessarily imply a GC influence.
To control for false discoveries, the well-established results on the asymptotic normality of maximum likelihood estimators can be utilized: under mild assumptions, $n{\mathcal{F}}_{y\mapsto x}$ converges in distribution to a chi-square $\chi_{p}^{2}$ with degree $p$.
In addition, under a sequence of local alternatives $H^n_{y\mapsto x,1} : {\boldsymbol{\theta}}_{(2)}^* = \boldsymbol{\delta} /\sqrt{n}$, for some constant vector $\boldsymbol{\delta}$, $n{\mathcal{F}}_{y\mapsto x}$ converges in distribution to a non-central chi-square $\chi_{p}^{2}(\nu)$ with degree $p$ and non-centrality $\nu > 0$ \cite{Davidson1970,Wald1943}.
These asymptotic results lead to a simple thresholding strategy: rejecting the null hypothesis if ${\mathcal{F}}_{y\mapsto x}$ exceeds a fixed threshold. A key consideration in this framework is choosing the model order $p$.
To this end, criteria such as the AIC \cite{Akaike1974model} and BIC \cite{schwarz1978} are widely used to strike a balance between the variance accounted for and the number of coefficients to be estimated.

While the foregoing procedure works well in practice for large sample sizes, its performance sharply degrades as the sample size decreases. This performance degradation has two main reasons:
\begin{enumerate}
\item The regression models become under-determined and result in poor estimates of the parameters{\cite{basu2015regularized}}, and 
\item The conventional model selection criteria fail to capture possible long-range temporal coupling of the underlying processes{\cite{kazemipour2017sampling}}.
\end{enumerate}
As a result, the classical GC measure is highly susceptible to over-fitting.
In addition, when the process noise elements $\epsilon_t$ and $\epsilon'_t$ are highly correlated, the OLS estimates incur additional error in capturing the true BVAR parameters, and hence result in mis-detection of the GC influences.
While some existing nonparametric methods aim at entirely bypassing MVAR estimation by utilizing spectral matrix factorization \cite{DHAMALA2008354} or multivariate embedding \cite{PhysRevE.82.016207} for system identification, they are similarly prone to the adverse effects of small sample size. {It is noteworthy that there also exists a slew of partial correlation-based nonparametric methods that employ conditional independence tests for GC detection \cite{runge2019inferring,runge2019detecting,runge2015identifying}, thus avoiding time series modeling assumptions altogether.} 
\subsection{LASSO-based GC Inference in the High-Dimensional Setting}
In the so-called high-dimensional setting, where the model dimension becomes comparable to or even exceeds the sample size {\cite{buhlmann2011statistics}}, regularization schemes are employed to guard against over-fitting.
These schemes include Tikhonov regularization \cite{Golub1999Tikhonov,Frank1984tikhonov}, $\ell_1$-regularization or the LASSO \cite{Tibshirani1996Lasso,loh2011high,negahban2012,hastie2015statistical}, smoothly clipped absolute deviation \cite{fan2001variable,xie2009}, Elastic-Net \cite{zou2005regularization}, and their variants, and have particularly proven useful in MVAR estimation \cite{Goldenshluger2001,wang2007regression,nardi2011autoregressive,han2013transition,SKRIPNIKOV2019164,basu2019low}.
Among these techniques, the LASSO has been widely used and studied in the high-dimensional sparse MVAR setting, under fairly general assumptions \cite{wang2007regression,nardi2011autoregressive,han2013transition}.
By augmenting the least squares error loss with the $\ell_1$-norm of the parameters, the LASSO simultaneously guards against over-fitting and provides automatic model selection \cite{zhao2006model,HSU20083645,REN20101705,ren2013model,Ding2018Model,hastie2015statistical}, under the hypothesis that the true parameters are sparse. In the context of MVAR estimation, assuming that the time series data admit a sparse MVAR representation, the LASSO estimates enjoy tight bounds on the estimation and prediction errors under suitable sample size requirements, even for models with correlated noise \cite{basu2015regularized,wong2016lasso}.

By leveraging the foregoing properties, the LASSO and its variants have been utilized in existing work to identify GC influences in a graphical fashion \cite{fujitaModelingGeneExpression2007,arnold2007temporal,Lozano2009GGGM,Shojaie2010GGC,liu2012sparse,Basu2015NGC,MICHAILIDIS2013326}.
{Fujita et al. \cite{fujitaModelingGeneExpression2007} used sparse autoregressive models to identify the structure of gene regulatory networks; Arnold et al. \cite{arnold2007temporal} provided a formal treatment of LASSO-based Granger causality detection using sparse parameter estimates. Lozano et al. \cite{Lozano2009GGGM} and Shimamura et al. \cite{shimamuraRecursiveRegularizationInferring2009} used Group LASSO and Elastic-Net penalties, respectively, to reduce the number of false
positives while maintaining a high true positive rate in network inference. Finally, Shojaie and Michailidis \cite{Shojaie2010GGC} introduced a truncating LASSO penalty in order to identify the correct order of the VAR model and thereby enhance the reliability of GC discovery. To deal with time series data with extreme events, i.e., exhibiting heavy-tailed distributions, Liu et al. \cite{liu2012sparse} utilized the theory of extreme value modeling to modify the distributional assumptions.}

{The aforementioned methods translate the non-zero values of estimated parameters (in a group-wise sense) to GC either directly or after appropriate pruning \cite{Basu2015NGC} (See \cite{MICHAILIDIS2013326} for a through overview). Alternatively, de-biasing techniques have been introduced for constructing confidence intervals over the estimated parameters and thereby identifying the significant GC effects \cite{tang2012measuring,van2014asymptotically,van2017efficiency,javanmard2014confidence,javanmard2018debiasing,javanmard2019false,zhangConfidenceIntervalsLow2014} to maintain high number of true positives while reducing false discoveries.}
\subsection{Unifying the Classical OLS-based and LASSO-based Approaches} \label{sec:unifying}
{Comparing the classical OLS-based GC and the recent LASSO-based approaches to GC inference reveals an evident disconnect: the latter approach directly utilizes the \emph{estimated parameters} from a \emph{single model} to identify the GC influence with non-asymptotic performance guarantees, while the former is based on comparing the \emph{prediction error} performance of \emph{two different models} by resorting to asymptotic distributions for statistical testing.} 

Our main objective here is {address this disconnect} by unifying these two approaches. To this end, we first replace OLS estimation in \prettyref{eq:ols} by its LASSO counterpart:
\begin{subequations}
\begin{align} 
\widehat{\boldsymbol{\theta}} & =\underset{\boldsymbol{\theta}}{\text{argmin}}\hspace*{2mm}\frac{1}{n}\left\Vert \mathbf{x}-\mathbf{X}\boldsymbol{\theta}\right\Vert ^{2}+\lambda_{n}\Vert\boldsymbol{\theta}\Vert_{1} \\
\widehat{\widetilde{\boldsymbol{\theta}}}&=\underset{\boldsymbol{\theta}:\boldsymbol{\theta}_{(2)}=\boldsymbol{0}}{\text{argmin}}\hspace*{2mm}\frac{1}{n}\left\Vert \mathbf{x}-\mathbf{X}\boldsymbol{\theta}\right\Vert ^{2}+\lambda_{n}\Vert\boldsymbol{\theta}\Vert_{1},
\end{align}
\label{eq:Lasso-full&reducedmoedel}
\end{subequations}%
where $\lambda_{n}$ denotes the regularization parameter. Let
\begin{align*}
\ell(\boldsymbol{\theta}_{(1)}, \boldsymbol{\theta}_{(2)}):=\frac{1}{n}\left\Vert \mathbf{x}-\mathbf{X}\boldsymbol{\theta}\right\Vert ^{2} \text{ with } \boldsymbol{\theta} = [\boldsymbol{\theta}_{(1)}; \boldsymbol{\theta}_{(2)}].
\end{align*}%
By similarly grouping the solutions of \prettyref{eq:Lasso-full&reducedmoedel} as $\widehat{\boldsymbol{\theta}} = [\widehat{\boldsymbol{\theta}}_{(1)}; \widehat{\boldsymbol{\theta}}_{(2)}]$ and $\widehat{\widetilde{\boldsymbol{\theta}}} = [\widehat{\widetilde{\boldsymbol{\theta}}}_{(1)}; \mathbf{0}]$, we then propose to use the following statistic:
\begin{align}
{\mathcal{T}}_{y\mapsto x}:=\frac{\ell\Big(\widehat{\widetilde{\boldsymbol{\theta}}}_{(1)},\boldsymbol{0}\Big)}{\ell\Big(\widehat{\boldsymbol{\theta}}_{(1)}, \widehat{\boldsymbol{\theta}}_{(2)}\Big)} - 1=\frac{\ell\Big(\widehat{\widetilde{\boldsymbol{\theta}}}_{(1)},\boldsymbol{0}\Big)-\ell\Big(\widehat{\boldsymbol{\theta}}_{(1)}, \widehat{\boldsymbol{\theta}}_{(2)}\Big)}{\ell\Big(\widehat{\boldsymbol{\theta}}_{(1)}, \widehat{\boldsymbol{\theta}}_{(2)}\Big)},\label{eq:T-statistic}
\end{align}%
akin to a scaled likelihood-based version of the F-statistic, which we call the LASSO-based GC (LGC) statistic. Note that the LGC statistic can be related to the conventional GC statistic as ${\mathcal{T}}_{y\mapsto x} = \exp({\mathcal{F}}_{y\mapsto x})-1$, when $\lambda_n = 0$. Therefore, it is expected for \({\mathcal{T}}_{y\mapsto x}\) to be near $0$ under the null hypothesis.
One advantage of using this statistic is that a simple thresholding strategy, similar to that used for the classical GC statistic, can be used to reject the null hypothesis $H_{y\mapsto x,0}:\boldsymbol{\theta}^*_{(2)}=\boldsymbol{0}$.
In the next section, we will characterize the non-asymptotic properties of the LGC statistic and seek conditions that allow us to distinguish between the null (i.e., absence of a GC effect) and a {suitably defined} alternative (i.e., presence of a GC effect) hypothesis. 

{To evaluate the benefits of unifying these two approaches, i.e., using the robust \emph{estimation} performance of the LASSO in comparing the \emph{prediction error} performance of the full and reduced models, we will present a  simulation study in \prettyref{sec:Simulation-Studies} to compare LGC-based GC identification with two existing LASSO-based approaches, namely, the confidence interval (CI)-based LASSO \cite{van2014asymptotically,javanmard2018debiasing}, and truncating LASSO (TLASSO). Our simulation results show that when the strength of the GC link is relatively weak, the CI-based LASSO method results in a high false positive rate and TLASSO exhibits low true positive rate; LGC, however, strikes a reasonable balance between true and false positive discovery, in accordance with our main objective.}
\section{Theoretical Results \label{sec:Main-contribution}}
Our main theoretical contribution in this section is to characterize $\mathcal{T}_{y \mapsto x}$ under both the null and a suitably chosen alternative hypothesis, and establish sufficient conditions that guarantee distinguishing these hypotheses with high probability. {Furthermore, we show that some of these conditions are indeed necessary and thus pose a fundamental limit for the separation of null and alternative hypotheses.} We then analyze the false positive error probability corresponding to the aforementioned thresholding strategy, under slightly weakened sufficient conditions. {The latter result can be used to obtain suitable thresholds in practice, as we will discuss in \prettyref{sec:discussion} and demonstrate in \prettyref{sec:experimental validation}.} Before presenting the main results, we state our key assumptions and {discuss} their implications in the following section.
\subsection{Key Assumptions and Their Implications}
We adopt the following assumption on the BVAR process from \cite{basu2015regularized}:
\begin{assum}
\label{assum:key}
The $\left\{x_{t},y_{t}\right\}_{t=-p+1}^{n}$ is a part {of a} realization of zero-mean bivariate process that admits a stable and invertible BVAR$(p)$ representation, with a zero-mean i.i.d. Gaussian process noise with positive definite covariance $\boldsymbol{\Sigma}_\epsilon$. Further, the initial condition of the process is such that the samples under consideration attain the stationary distribution.
\end{assum}
To elaborate on the implications of \prettyref{assum:key} for the second order statistics of the BVAR process, let $\boldsymbol{\Gamma}(l)$ be the auto-covariance matrix of the BVAR process at lag $l$ and  
\begin{align*}
\mathbf{F}(\omega):= \frac{1}{2\pi}\sum_{l=-\infty}^{\infty}\boldsymbol{\Gamma}(l)\exp\left(-il\omega\right)
\end{align*}
be its spectral density matrix. It can be shown that the spectral density exists, if \(\sum_{l=0}^{\infty} \Vert \boldsymbol{\Gamma}(l) \Vert_2^2 < \infty\). Furthermore, if \(\sum_{l=0}^{\infty} \Vert \boldsymbol{\Gamma}(l) \Vert_2 < \infty\), the spectral density is bounded and continuous, so that the essential supremum is indeed achieved \cite{basu2015regularized}.

For the BVAR$(p)$ process in \prettyref{eq:true-MVAR(p)}, the matrix valued characteristic polynomial is defined as \(\mathbf{A}(z) := \mathbf{I} - \sum_{j=1}^{p}\mathbf{A}_jz^j\). Then, the following consequences of \prettyref{assum:key} provide a simple characterization of the spectral density matrix \cite{basu2015regularized}:
\begin{enumerate}
  \item The process noise covariance matrix \(\boldsymbol{\Sigma}_{\epsilon}\) is positive definite with bounded eigen-values, i.e., 
  \begin{align*}
  0<\Lambda_{\min}(\boldsymbol{\Sigma}_{\epsilon})\leq\Lambda_{\max}(\boldsymbol{\Sigma}_{\epsilon})<\infty.
  \end{align*}
  \item The BVAR process is stable and invertible, i.e., \(\operatorname{det}(\mathbf{A}(z)) \neq 0\) on or inside the unit 
  circle, \(\{z \in \mathbbm{C}: \vert z \vert \leq 1\}\).
\end{enumerate}
Under these two conditions, the spectral density matrix exists, and its maximum eigen-value is bounded almost everywhere on
$[-\pi,\pi]$, i.e., 
\begin{align*}
\mathcal{M}(\mathbf{F}):=\underset{\omega\in[-\pi,\pi]}{\text{ess sup}}\Lambda_{\max}(\mathbf{F}(\omega))<\infty.
\end{align*}
Furthermore, it is bounded and continuous, and admits the representation:
\begin{align*}\mathbf{F}(\omega) = \frac{1}{2\pi} \mathbf{A}^{-1}(\exp{(-i\omega)}) \boldsymbol{\Sigma}_{\epsilon} \mathbf{A}^{-H}(\exp{(-i\omega)}).\end{align*} 
Additionally, consider the infimum of the spectral density over unit circle:
\begin{align}
\mathscr{m}(\mathbf{F}):=\underset{\omega\in[-\pi,\pi]}{\text{ess inf}}\Lambda_{{\min}}(\mathbf{F}(\omega)).
\end{align}
Then, the following useful bounds hold for the BVAR$(p)$ process in \prettyref{eq:true-MVAR(p)} \cite{basu2015regularized}:
\begin{align}\label{eq:M-bounds}
\mathcal{M}(\mathbf{F}) \leq \frac{1}{2\pi}\frac{\Lambda_{\max}(\boldsymbol{\Sigma}_{\epsilon})}{\mu_{\min}(\mathbf{A})}, \hspace{2em} \mathscr{m}(\mathbf{F}) \geq \frac{1}{2\pi}\frac{\Lambda_{\min}(\boldsymbol{\Sigma}_{\epsilon})}{\mu_{\max}(\mathbf{A})},
\end{align}
where \(\mu_{\max}\left(\mathbf{A}\right):=\underset{\vert z\vert=1}{\text{max}} \Lambda_{\max}\left(\mathbf{A}^H(z)\mathbf{A}(z)\right) \text{ and } \mu_{\min}\left(\mathbf{A}\right):=\underset{\vert z\vert=1}{\text{min}} \Lambda_{\min}\left(\mathbf{A}^H(z)\mathbf{A}(z)\right).\)

We note that the characteristic polynomial \(\mathbf{A}(z)\) encodes the temporal dependencies of the process, whereas \(\boldsymbol{\Sigma}_{\epsilon}\) captures the correlation between the process noise components, possibly due to latent processes. 
Expressing the error bounds in our theoretical analysis in terms of \({\mu_{\max}(\mathbf{A}), \mu_{\min}(\mathbf{A})},{\Lambda_{\max}(\boldsymbol{\Sigma}_{\epsilon}), \Lambda_{\min}(\boldsymbol{\Sigma}_{\epsilon})}\), instead of \(\mathcal{M}(\mathbf{F})\) and \(\mathscr{m}(\mathbf{F})\), helps to distinguish the contributions of these two sources of BVAR dependencies \cite{basu2015regularized} (See Appendix \ref{sec:prediction-error-analysis}).

We also consider the $2p$-dimensional alternative BVAR$(1)$ representation of the $2$-dimensional BVAR$(p)$ process: \(\mathbf{X}_t = \breve{\mathbf{A}}_1\mathbf{X}_{t-1} + \breve{\boldsymbol{\epsilon}}\), where \(\mathbf{X}_t\) is the first the row of \(\mathbf{X}\) in \prettyref{eq:regressors} organized as a column vector, and $\breve{\mathbf{A}}_1$ and $\breve{\boldsymbol{\epsilon}}$ are constructed by the corresponding augmentation of $\mathbf{A}_i$'s and $\epsilon_t$'s, respectively.
The process \(\mathbf{X}_t\) has a characteristic polynomial, \(\breve{\mathbf{A}}(z) := \mathbf{I} - \breve{\mathbf{A}}_1z\) and is stable if and only if the original process is stable \cite{helmut2005new}.

The remaining component of our key assumptions is the following sparsity assumption, frequently arising in high-dimensional regime, specially in LASSO literature.
\begin{assum}
The regression coefficients in \prettyref{eq:MVAR(p)}, \(\boldsymbol{\theta}^{*}\) is $k$-sparse, i.e.,
$\Vert\boldsymbol{\theta}^{*}\Vert_{0}=k$. \label{assum:sparsity}
\end{assum} 
The implication of this assumption for the \emph{full} model is fairly standard, under both the null and alternative hypotheses.
{However, for the \emph{reduced} model under the alternative hypothesis, where only the autoregressive parameters are unspecified and the cross-regression parameters are enforced to be $\boldsymbol{0}$, we need to define a suitable surrogate ``true'' model whose parameters can be used to quantify the estimation error and thereby establish concentration bounds (See \prettyref{prop:prediction error-reducedmodel}). Let the columns of $\mathbf{X}$ corresponding to $\boldsymbol{\theta}_{(i)}$ be denoted by $\mathbf{X}_{(i)}$, for $i=1,2$.
Using the fact that the error residuals of the optimal linear estimator, in the mean square error sense, are uncorrelated with the columns of the design matrix \cite[pp. 386]{kayFundamentalsStatisticalSignal1993}), we define the surrogate ``true'' autoregression coefficients in the \emph{reduced} model as:}
\begin{align}\widetilde{\boldsymbol{\theta}}^*_{(1)} :=\boldsymbol{\theta}^*_{(1)}+\mathbf{C}_{11}^{-1}\mathbf{C}_{12}{\boldsymbol{\theta}^*_{(2)}},\end{align}
where
\begin{align}
\nonumber \mathbbm E\left[\frac{1}{n}\mathbf{X}^{\top}\mathbf{X}\right] &=\left[\begin{alignedat}{1}\mathbbm{E}\left[\frac{1}{n}\mathbf{X}_{(1)}^{\top}\mathbf{X}_{(1)}\right] & \mathbbm{E}\left[\frac{1}{n}\mathbf{X}_{(1)}^{\top}\mathbf{X}_{(2)}\right]\\
\mathbbm{E}\left[\frac{1}{n}\mathbf{X}_{(2)}^{\top}\mathbf{X}_{(1)}\right] & \mathbbm{E}\left[\frac{1}{n}\mathbf{X}_{(2)}^{\top}\mathbf{X}_{(2)}\right]
\end{alignedat}
\right]\\
&=:\begin{bmatrix}\mathbf{C}_{11} & \mathbf{C}_{12}\\
\mathbf{C}_{21} & \mathbf{C}_{22}
\end{bmatrix}=:\mathbf{C}.
\end{align}
Note that even though the MVAR coefficients under the alternative hypothesis are $k$-sparse, the surrogate "true" autoregression coefficients under the \emph{reduced} model may not be.
To deal with this issue, we follow the treatment of \cite{negahban2012} in analyzing the LASSO under weakly sparse or compressible parameters, and further impose a norm condition on $\boldsymbol{\theta}^*_{(2)}$ (See the statement of \prettyref{thm:main-theorem}) to restrict the alternative hypothesis. The latter ensures that the \emph{full} and \emph{reduced} models are distinguishable under the alternative hypothesis.

Finally, we use the following definition throughout the rest of the paper:
\begin{definition}
We say that a threshold $\mathscr{t}$ \emph{correctly distinguishes} the null and alternative hypotheses, if the ranges of the statistic $\mathcal{T}$ under the null and alternative hypotheses are respectively subsets of $\mathcal{T} < \mathscr{t}$ and $\mathcal{T} > \mathscr{t}$.
\label{def:def}
\end{definition}

\subsection{Main Theoretical Results}
Our main theorem is stated as follows:
\begin{thm}[Main Theorem]
\label{thm:main-theorem} Suppose that the key assumptions \ref{assum:key} and \ref{assum:sparsity} hold. Then, for the proposed LGC statistic $\mathcal{T}_{y \mapsto x}$ evaluated at the BVAR($p$) parameter estimates from the solutions of \prettyref{eq:Lasso-full&reducedmoedel} with a regularization parameter $\lambda_{n}=4\mathscr{A}\sqrt{{\log(2p)}/{n}}$, there exists a threshold that correctly distinguishes the null and the local alternative hypothesis $H^n_{y\mapsto x,1} : \| {\boldsymbol{\theta}^*_{(2)}}\|_2^2 \geq \mathscr{B}{k\log(2p)}/{n}$ with probability at least {$1 - K_1\exp\left(-n\bar{c}\right) - {K_2}/{p^{\bar{d}}}$}, if {$n\geq \max\{\mathscr{C}'', \mathscr{D}'' k\} \log(2p)$},
with $\mathscr{A}$, $\mathscr{B}$, $\mathscr{C}''$, $\mathscr{D}''$, $K_1, K_2, \bar{c},$ and $\bar{d} > 0$ denoting constants that are explicitly given in the proof.
\end{thm}
\begin{proof}
The proof has three main steps. First, we bound the deviation of the empirical quantities $\ensuremath{\ell(\widehat{\boldsymbol{\theta}}_{(1)},\widehat{\boldsymbol{\theta}}_{(2)})}$ (\emph{full} model) and $\ell(\widehat{\widetilde{\boldsymbol{\theta}}}_{(1)},\boldsymbol{0})$ (\emph{reduced} model) with respect to their counterparts evaluated at the true parameters. After invoking suitable concentration results for $\ensuremath{\ell({\boldsymbol{\theta}}_{(1)}^{*},{\boldsymbol{\theta}}_{(2)}^{*})}$ and $\ell(\widehat{{\boldsymbol{\theta}}}_{(1)}^{*},\boldsymbol{0})$, we can then lower bound $\mathcal{T}_{y\mapsto x}$ under the alternative hypothesis and upper bound it under the null hypothesis. Secondly, we seek conditions under which the bounds do not coincide, which further restricts the alternative hypothesis. The last step of the proof establishes that these deviation and concentration results indeed hold with high probability.

\noindent {\bf Step 1.} We first assume the deviation bounds:
\begin{align}
\left\vert \ell\Big(\widehat{\boldsymbol{\theta}}_{(1)}, \widehat{\boldsymbol{\theta}}_{(2)}\Big)-\ell\Big(\boldsymbol{\theta}^{*}_{(1)}, \boldsymbol{\theta}^*_{(2)}\Big)\right\vert &\leq \Delta_{F}, \label{eq:full dev}\\
\left\vert \ell\Big(\widehat{\widetilde{\boldsymbol{\theta}}}_{(1)},\boldsymbol{0}\Big)-\ell\left(\widetilde{\boldsymbol{\theta}}^*_{(1)},\boldsymbol{0}\right)\right\vert &\leq \Delta_{R}, \label{eq:reduced dev}
\end{align}
and the concentration {inequalities}:
\begin{align}
\left\vert\ell\Big(\boldsymbol{\theta}^{*}_{(1)}, \boldsymbol{\theta}^*_{(2)}\Big) - (\boldsymbol{\Sigma}_\epsilon)_{1,1}\right\vert &\leq \Delta_{N} \label{eq:conct 1}\\
\left\vert \ell\left(\widetilde{\boldsymbol{\theta}}^*_{(1)},\boldsymbol{0}\right)-\ell\Big(\boldsymbol{\theta}^{*}_{(1)}, \boldsymbol{\theta}^*_{(2)}\Big) - D\right\vert &\leq \Delta_{D} \label{eq:conct 2}
\end{align}
hold for some non-negative quantities, \(\Delta_{F}\), \(\Delta_{R}\), \(\Delta_{N}\), \(\Delta_{D}\) and \(D:={\boldsymbol{\theta}^*_{(2)}}^{\top}(\mathbf{C}_{22}-\mathbf{C}_{21}\mathbf{C}_{11}^{-1}\mathbf{C}_{12}){\boldsymbol{\theta}^*_{(2)}}.\)
Using the bounds in \eqref{eq:full dev} and \eqref{eq:reduced dev}, we get: 
\begin{align}
\frac{\ell\left(\widetilde{\boldsymbol{\theta}}^*_{(1)},\boldsymbol{0}\right)-\Delta_{R}}{\ell\Big(\boldsymbol{\theta}^{*}_{(1)}, \boldsymbol{\theta}^*_{(2)}\Big)+\Delta_{F}} \leq\frac{\ell\Big(\widehat{\widetilde{\boldsymbol{\theta}}}_{(1)},\boldsymbol{0}\Big)}{\ell\Big(\widehat{\boldsymbol{\theta}}_{(1)}, \widehat{\boldsymbol{\theta}}_{(2)}\Big)}\leq\frac{\ell\left({\widetilde{\boldsymbol{\theta}}}^*_{(1)},\boldsymbol{0}\right)+\Delta_{R}}{\ell\Big(\boldsymbol{\theta}^{*}_{(1)}, \boldsymbol{\theta}^*_{(2)}\Big)-\Delta_{F}},
\end{align}
which gives the following lower and upper bounds on $\mathcal{T}_{y \mapsto x}$:
\begin{align}
\nonumber & \frac{\ell\left(\widetilde{\boldsymbol{\theta}}^*_{(1)},\boldsymbol{0}\right)-\ell\Big(\boldsymbol{\theta}^{*}_{(1)}, \boldsymbol{\theta}^*_{(2)}\Big)-\Delta_{R}-\Delta_{F}}{\ell\Big(\boldsymbol{\theta}^{*}_{(1)}, \boldsymbol{\theta}^*_{(2)}\Big)+\Delta_{F}}\leq\mathcal{T}_{y \mapsto x} \\
& \qquad \qquad \qquad \leq \frac{\ell\left(\widetilde{\boldsymbol{\theta}}^*_{(1)},\boldsymbol{0}\right)-\ell\Big(\boldsymbol{\theta}^{*}_{(1)}, \boldsymbol{\theta}^*_{(2)}\Big)+\Delta_{R}+\Delta_{F}}{\ell\Big(\boldsymbol{\theta}^{*}_{(1)}, \boldsymbol{\theta}^*_{(2)}\Big)-\Delta_{F}}.
\end{align}
Now, under the null hypothesis $H_{y\mapsto x,0}:\boldsymbol{\theta}^{{*}}_{(2)}=\boldsymbol{0}$,
 we have $\ell\left(\widetilde{\boldsymbol{\theta}}^*_{(1)},\boldsymbol{0}\right)=\ell\Big(\boldsymbol{\theta}^{*}_{(1)}, \boldsymbol{\theta}^*_{(2)}\Big)$ and {thus $\Delta_R$ in Proposition \ref{prop:reduced model deviation} can be chosen equal to $\Delta_F$ in Proposition \ref{prop:full model deviation}.} This implies:
\begin{alignat}{1}
\mathcal{T}_{y \mapsto x} & \leq\frac{\Delta_{R}+\Delta_{F}}{\ell\Big(\boldsymbol{\theta}^{*}_{(1)}, \boldsymbol{\theta}^*_{(2)}\Big)-\Delta_{F}}\leq\frac{2\Delta_{F}}{(\boldsymbol{\Sigma}_\epsilon)_{1,1}-\Delta_{N}-\Delta_{F}},\label{eq:ub null}
\end{alignat}
with application of \eqref{eq:conct 1}.

On the other hand, under a general alternative hypothesis $H_{y\mapsto x,0}:\boldsymbol{\theta}^*_{(2)}\neq\boldsymbol{0}$, \prettyref{eq:conct 2} can be used to show that: 
\begin{alignat}{1}
\nonumber \mathcal{T}_{y \mapsto x} &\geq\frac{\ell\left(\widetilde{\boldsymbol{\theta}}^*_{(1)},\boldsymbol{0}\right)-\ell\Big(\boldsymbol{\theta}^{*}_{(1)}, \boldsymbol{\theta}^*_{(2)}\Big)-\Delta_{R}-\Delta_{F}}{\ell\Big(\boldsymbol{\theta}^{*}_{(1)}, \boldsymbol{\theta}^*_{(2)}\Big)+\Delta_{F}}\\
&\geq\frac{D-\left(\Delta_{D}+\Delta_{R}+\Delta_{F}\right)}{(\boldsymbol{\Sigma}_\epsilon)_{1,1}+\Delta_{N}+\Delta_{F}}.\label{eq:lb non null}
\end{alignat}

From the upper and lower bounds in \prettyref{eq:ub null} and \prettyref{eq:lb non null}, it is possible to choose a threshold to correctly distinguish the two hypotheses, if: 
\begin{align}\frac{D-\left(\Delta_{D}+\Delta_{R}+\Delta_{F}\right)}{(\boldsymbol{\Sigma}_\epsilon)_{1,1}+\Delta_{N}+\Delta_{F}}>\frac{2\Delta_{F}}{(\boldsymbol{\Sigma}_\epsilon)_{1,1}-\Delta_{N}-\Delta_{F}},
\end{align}
which after rearrangement translates to: 
\begin{align}
D>\Delta_{D}+\Delta_{R}+\Delta_{F}\left(1+2\frac{(\boldsymbol{\Sigma}_\epsilon)_{1,1}+\Delta_{N}+\Delta_{F}}{(\boldsymbol{\Sigma}_\epsilon)_{1,1}-\left(\Delta_{N}+\Delta_{F}\right)}\right) & .\label{eq:thresholding_cond}
\end{align}
Next, we choose $\Delta_N = {(\boldsymbol{\Sigma}_\epsilon)_{1,1}}/{4}$. Then, assuming $\Delta_{F} \le {(\boldsymbol{\Sigma}_\epsilon)_{1,1}}/{4}$, the bound of \prettyref{eq:thresholding_cond} further simplifies to
\begin{align}
D>\Delta_{D}+\Delta_{R}+7\Delta_{F}. \label{eq:thresholding cond simplified}
\end{align}
\noindent {\bf Step 2.} 
Next, we assume that the following conditions hold:
\begin{cond}[Restricted eigenvalue (RE) condition] The symmetric matrix
$\widehat{\boldsymbol{\Sigma}}=\mathbf{X}^{\top}\mathbf{X}/n\in\mathbbm R^{2p\times2p}$
satisfies restricted eigenvalue condition with curvature $\alpha>0$
and tolerance $\tau\ge0$, i.e., $\widehat{\boldsymbol{\Sigma}}\sim\text{RE}(\alpha,\tau)$:
\begin{align}
\nonumber \boldsymbol{\phi}^{\top}\widehat{\boldsymbol{\Sigma}}\boldsymbol{\phi}\geq\alpha\Vert\boldsymbol{\phi}\Vert_2^{2}-\tau\Vert\boldsymbol{\phi}\Vert_{1}^{2},\hspace*{2mm}\forall\hspace*{2mm}\boldsymbol{\phi}\in\mathbbm R^{2p},
\end{align}
with 
\begin{align*}\tau := \frac{m-1}{m}\frac{\alpha}{32 k}\end{align*}
for some constant $m > 1$. \label{cond:RE}
\end{cond}
\begin{cond}[Deviation condition]
There exist deterministic functions $\mathbbm Q(\boldsymbol{\theta}^*,{\boldsymbol{\Sigma}}_\epsilon)$, $\mathbbm Q'(\boldsymbol{\theta}^*,{\boldsymbol{\Sigma}}_\epsilon)$ such that 
\begin{align*}
&\left\Vert \frac{1}{n}\mathbf{X}^{\top}({\mathbf{x}-\mathbf{X}\boldsymbol{\theta}^{*}})\right\Vert _{\infty}\leq \mathbbm Q(\boldsymbol{\theta}^*,{\boldsymbol{\Sigma}}_\epsilon)\sqrt{\frac{\log(2p)}{n}}, \\
&\left\Vert \frac{1}{n}\mathbf{X}_{(1)}^{\top}\left({\mathbf{x}-\mathbf{X}_{(1)}{{\widetilde{\boldsymbol{\theta}}}^*_{(1)}}}\right)\right\Vert _{\infty}\leq \mathbbm Q'(\boldsymbol{\theta}^*,{\boldsymbol{\Sigma}}_\epsilon)\sqrt{\frac{\log(2p)}{n}}.
\end{align*} \label{cond:DB}
\end{cond}
\noindent Under these conditions, we can use the expressions for the non-negative quantities, \(\Delta_{F}\), \(\Delta_{R}\) and \(\Delta_{D}\) derived in Propositions \ref{prop:full model deviation} and \ref{prop:reduced model deviation} (Appendix \ref{sec:prediction-error-analysis}) and \prettyref{lem:reducedmodel-fullmodel_error} (Appendix \ref{sec:Concentration-Lemmas}), respectively, to obtain the following bound on the right hand side of \prettyref{eq:thresholding cond simplified}:
\begin{align*}
\Delta_{D}+\Delta_{R}+7\Delta_{F} \leq& \, \mathscr{a}\sqrt{\frac{\log(2p)}{n}}\left\Vert {\boldsymbol{\theta}^*_{(2)}}\right\Vert_2^{2}\\
& +\mathscr{b}\sqrt{\frac{k\log(2p)}{n}}\left\Vert {\boldsymbol{\theta}^*_{(2)}}\right\Vert_2 + \mathscr{c}\frac{k\log(2p)}{n},
\end{align*}
where
\begin{align*}
\mathscr{a} & = \frac{\alpha}{27}\left(\left\Vert \mathbf{C}_{11}^{-1}\mathbf{C}_{12}\right\Vert_2^2 + 1\right),\\
\mathscr{b} & = \mathscr{A}\left[\left(32\sqrt{2m} + 73\right)\left\Vert \mathbf{C}_{11}^{-1}\mathbf{C}_{12}\right\Vert + 1\right],\\
\mathscr{c} & =\frac{16\mathscr{A}^{2}}{\alpha/m}\left(\frac{168}{m+1}+20\right),
\end{align*}
provided the LASSO problems are solved with {the} choice of $\lambda_{n}=4\mathscr{A}\sqrt{{\log(2p)}/{n}}$, for $\mathscr{A}$ satisfying (See Propositions \ref{prop:deviation condition} and \ref{prop:deviation condition-2} in Appendix \ref{sec:verif}):
\begin{align*}
\mathscr{A}\geq\max\left\{\mathbbm Q(\boldsymbol{\theta}^{*},\boldsymbol{\Sigma}_{\epsilon}),\mathbbm Q'(\boldsymbol{\theta}^{*},\boldsymbol{\Sigma}_{\epsilon})\right\}.
\end{align*}
Also, note that the assumption \(\Delta_{F} \le {(\boldsymbol{\Sigma}_\epsilon)_{1,1}}/{4}\) requires:
\begin{equation}
\frac{24}{m+1}\frac{k\lambda_{n}^{2}}{\alpha/m} \le \frac{(\boldsymbol{\Sigma}_\epsilon)_{1,1}}{4}.
\label{eq:delta F}
\end{equation}
and imposes an upper bound on $\lambda_n$. The implications of this upper bound are further discussed in \emph{Remark \ref{rem:1}} at the end of this section.

Next, we use the following lower bound on $D$:  $D\geq\widetilde{\Lambda}_{\min}\big\Vert {\boldsymbol{\theta}^*_{(2)}}\big\Vert_2^{2}$,
with $\widetilde{\Lambda}_{\min}:=\Lambda_{\min}\left(\mathbf{C}_{22}-\mathbf{C}_{21}\mathbf{C}_{11}^{-1}\mathbf{C}_{12}\right)$, which gives the following sufficient condition for inequality \prettyref{eq:thresholding cond simplified} to hold:
\begin{align}
\nonumber \tilde{\Lambda}_{\min}\left\Vert {\boldsymbol{\theta}^*_{(2)}}\right\Vert_2^{2}\geq & \, \mathscr{a}\sqrt{\frac{\log(2p)}{n}}\left\Vert {\boldsymbol{\theta}^*_{(2)}}\right\Vert_2^{2} +\mathscr{b}\sqrt{\frac{k\log(2p)}{n}}\left\Vert {\boldsymbol{\theta}^*_{(2)}}\right\Vert_2\\
& +\mathscr{c}\frac{k\log(2p)}{n}.\label{eq:cond 2}
\end{align}
By further requiring 
\(n \geq \mathscr{C}' \log(2p),\)
where
\begin{align*}
\mathscr{C}' = \left(\frac{2\alpha\left(\left\Vert \mathbf{C}_{11}^{-1}\mathbf{C}_{12}\right\Vert_2^2 + 1\right)}{27\widetilde{\Lambda}_{\min}}\right)^2,
\end{align*} 
we have $\widetilde{\Lambda}_{\min}-\mathscr{a}\sqrt{\log(2p)/n}\geq\widetilde{\Lambda}_{\min}/2$. The latter combined with \prettyref{eq:cond 2}, and an application of \prettyref{lem:useful-lemma} (Appendix \ref{sec:Concentration-Lemmas}) gives the sufficient condition:
\(
\big\Vert {\boldsymbol{\theta}^*_{(2)}}\big\Vert_2^{2} \geq \mathscr{B}{k\log(2p)}/{n}
\)
for unambiguous discrimination {of the LGC statistic under} the null and the \emph{local} alternative hypothesis $H^n_{y\mapsto x,0}: \big\Vert {\boldsymbol{\theta}^*_{(2)}}\big\Vert^{2}_2 \geq \mathscr{B}{k\log(2p)}/{n}$, as long as \(n\geq \max \{ \mathscr{C}', \mathscr{D}'k\} \log(2p)\),
where
\begin{align*}
\mathscr{B} := \left(\frac{4\mathscr{b}^{2}}{\widetilde{\Lambda}_{\min}^{2}}+\frac{4\mathscr{c}}{\widetilde{\Lambda}_{\min}}\right),
\text{ and } \mathscr{D}' := \frac{1536m}{m+1}\frac{\mathscr{A}^2}{\alpha(\boldsymbol{\Sigma}_\epsilon)_{1,1}}.
\end{align*}
Note that the condition $n \ge \mathscr{D}'k \log (2p)$ ensures the upper bound \prettyref{eq:delta F} on $\lambda_n$.

\noindent {\bf Step 3.} It only remains to provide a lower bound on the probability of the event that the proposed LGC statistic under the null and \emph{local} alternative hypotheses are correctly distinguishable by a threshold.
\prettyref{prop:RE condition} (Appendix \ref{sec:verif}) establishes that \prettyref{cond:RE} holds with probability at least
\begin{equation}
1-c_{1}\exp(-c_{2}n\min\{\zeta^{-2},1\}),
\end{equation}
if $n \ge C_0 \max\{\zeta^{2},1\}k\log(2p)$, for some constants $C_0, c_1, c_2$, and $\zeta$ $(>0)$.
Also, Propositions \ref{prop:deviation condition} and \ref{prop:deviation condition-2} (Appendix \ref{sec:verif}) establish that \prettyref{cond:DB} holds with probability at least
\begin{equation}
1 - \frac{d_{1}}{(2p)^{d_2}} - \frac{d'_{1}}{(2p)^{d'_2}},
\end{equation}
if $n \ge \max\{D_0,D_0'\}\log(2p)$, for some constants $d_1,d'_1,d_2$, $d'_2$, $D_0$, and $D'_0$ $(>0)$.

From \prettyref{lem:fullmodel_error} (Appendix \ref{sec:Concentration-Lemmas}), the first deviation bound \eqref{eq:full dev} holds with probability
\(1-2\exp\left(-{n}/{128}\right)\)
with the choice $\Delta_N = {(\boldsymbol{\Sigma}_\epsilon)_{1,1}}/{4}$. Lastly, from \prettyref{lem:reducedmodel-fullmodel_error} (Appendix \ref{sec:Concentration-Lemmas}), the deviation bound \eqref{eq:full dev} holds with probability
\(1-c_{3}\exp(-c_{4}n\min\{\zeta^{-2}\log(2p)/n,1\})\), under \prettyref{cond:DB}.

Combining the two steps, the claim of the theorem holds with probability at least
\begin{align}
\nonumber 1 &- 2\exp\left(-\frac{n}{128}\right) - c_{1}\exp(-c_{2}n\min\{\zeta^{-2},1\})\\
\nonumber & - c_{3}\exp\left(-c_{4}n\min\{\zeta^{-2}\max\{D_0, D_0^{'}\},1\}\right)\\
& - \frac{d_{1}}{(2p)^{d_2}} - \frac{d'_{1}}{(2p)^{d'_2}},
\label{eq:totalprob}
\end{align}
if $n \ge \max \{\mathscr{C}'', \mathscr{D}'' k\} \log (2p)$, where $\mathscr{D}'' := \max \{ \mathscr{D}', C_0 \max \{ \zeta^2,1\}\}$ and \(\mathscr{C}'' = \max \{\mathscr{C}',D_0,D_0'\} \).
Finally, this probability can be lower bounded by
\begin{align*} 
1 - K_1\exp\left(-n\bar{c}\right) - \frac{K_2}{p^{\bar{d}}},
\end{align*}
where
\begin{align*}
&\bar{c} = \min\left\{\frac{1}{128}, c_2, c_4, c_2\zeta^{-2}, c_4\zeta^{-2}\max\{D_0, D_0^{'}\}\right\},\\
&\bar{d} = \min\left\{d_2, d_2^'\right\},\\
& K_1 = 2 + c_1 + c_3, \hspace*{2em} K_2 = \frac{d_1}{2^{d_2}} + \frac{d'_1}{2^{d'_2}}.
\end{align*}
This concludes the proof of the main theorem.
\end{proof}

From the proof of \prettyref{thm:main-theorem}, it is apparent that large enough sample size $n$ is sufficient for the LASSO estimates under both the full and reduced models to be consistent given the choice of $\lambda_n$, which in turn allows us to control the prediction errors. Also, large enough $\|{\boldsymbol{\theta}^*_{(2)}}\|_2$ is sufficient to distinguish the null and alternative hypotheses. 

In order to establish the fundamental limit for separating the null and alternative hypotheses, however, we need to inspect the necessity of these conditions. From the proof of \prettyref{thm:main-theorem}, for any threshold $\mathscr{t}$, there exists a constant $\mathscr{A}'$, such that if $\lambda_n < 4 \mathscr{A}' \sqrt{\log(2p)/n}$, then $\mathcal{T}_{y \mapsto x} > \mathscr{t}$ with high probability, which results in a type I error. Thus, the choice of $\lambda_n$ in Theorem 1 is necessary to control false positives (See Remark \ref{rem:1} in \prettyref{sec:discussion}).

In what follows, we focus on the necessity of the lower bound on $\|{\boldsymbol{\theta}^*_{(2)}}\|_2$ and the condition on $n$ for detecting a Granger causal link. To prove this necessity result, we need a lower bound on the performance of the LASSO under the alternative hypothesis. To obtain a lower bound, we take the approach of \cite{goldenshluger2001nonasymptotic,kazemipour2017sampling} by constructing a family of BVAR processes with $k$-sparse parameters $\boldsymbol{\theta}_{(2)}$ for which the LASSO makes significant errors under the alternative hypothesis. The remaining key element of the proof is the Fano's inequality, which provides the required lower bound (See \prettyref{lem:ar_fano_ineq}):

\begin{thm}[Necessary Conditions]
\label{thm:ness cond} Suppose that the key assumptions \ref{assum:key} and \ref{assum:sparsity} hold. Then, there exists a constant $\widetilde{\mathscr{B}}$ for which the proposed LGC statistic $\mathcal{T}_{y \mapsto x}$ in \prettyref{thm:main-theorem} fails to reject the null hypothesis for any threshold $\mathscr{t}$ under the alternative hypothesis $H^n_{y\mapsto x,1} : \| {\boldsymbol{\theta}^*_{(2)}}\|_2^2 \le \widetilde{\mathscr{B}}{k\log p }/{n}$, with probability at least $1/2$, if {$n <  \widetilde{\mathscr{D}}' k \log p$}. The constants $\widetilde{\mathscr{B}}$ and $\widetilde{\mathscr{D}}'$ are explicitly given in the proof.
\end{thm}
\begin{proof}
Consider a class $\mathcal{Z}$ of BVAR processes with a fixed $1$-sparse $\boldsymbol{\theta}_{(1)}$ and $k$-sparse $\boldsymbol{\theta}_{(2)}$ over any subset $\mathcal{K} \subset \{1,2,\cdots,p\}$ satisfying $|\mathcal{K}| =k$, with elements given by
\begin{equation}
\label{eq:ar_minimax_params}
\left(\boldsymbol{\theta}_{(2)}\right)_{\ell} = \pm e^{-b}, \forall \ell \in \mathcal{K}.
\end{equation}
where $b$ is a constant to be chosen. We also assume, without loss of generality, that $\left(\boldsymbol{\theta}_{(1)}\right)_{1} = e^{-b}$ or $0$. Note that the vector $\boldsymbol{\theta} = [\boldsymbol{\theta}_{(1)}, \boldsymbol{\theta}_{(2)}]$ is $(k+1)$-sparse, but the following arguments can be repeated by redefining $k$ as $k-1$, pertaining to $k$-sparse parameter vectors. We also add the vector of all zeros to $\mathcal{Z}$. For a fixed $\mathcal{K}$, we have $2^{k+1}+1$ such parameters forming a subfamily $\mathcal{Z}_{\mathcal{K}}$. Consider the maximal collection of $\binom{p}{k}$ subsets $\mathcal{K}$ for which any two subsets differ in at least $k/4$ indices. The size of this collection can be identified by $A(p, \frac{k}{4}, k)$ in coding theory, where $A(n,d,w)$ represents the maximum size of a binary code of length $n$ with minimum distance $d$ and constant weight $w$ \cite{macwilliams1977theory}. It can be shown that
\begin{align*}
A(p, {\textstyle \frac{k}{4} }, k) \ge \frac{p^{\frac{7}{8}k - 1}}{k!},
\end{align*}
for large enough $p$ \cite[Theorem 6]{graham1980lower}. Also, by the Gilbert-Varshamov bound \cite{macwilliams1977theory}, there exists a subfamily $\mathcal{Z}_{\mathcal{K}}^\star \subset \mathcal{Z}_{\mathcal{K}}$, of cardinality $|\mathcal{Z}_{\mathcal{K}}^\star| \geq 2^{\lfloor k/8 \rfloor + 1}+1$, such that any two distinct $\boldsymbol{\theta}^i,\boldsymbol{\theta}^j \in \mathcal{Z}_{{\mathcal{K}}}^\star$ differ in at least $k/16$ elements. Thus for $\boldsymbol{\theta}^i, \boldsymbol{\theta}^j \in \mathcal{Z}^\star :=  \resizebox{!}{0.3cm}{$\displaystyle \bigcup_{\mathcal{K}}$} \mathcal{Z}^\star_{\mathcal{K}}$, we have
\vspace{-.2cm}
\begin{equation}
\label{eq:ar_theta_lowerbd}
\left \|\boldsymbol{\theta}^i-\boldsymbol{\theta}^j \right \|_2 \geq \frac{1}{4} \sqrt{k} e^{-b} =: h,
\end{equation}
and $\vert \mathcal{Z}^\star \vert \ge \frac{p^{\frac{7}{8}k - 1}}{k!} 2^{\lfloor k/8 \rfloor + 1} + 1$. For an arbitrary estimate $\widehat{\boldsymbol{\theta}} =: [\widehat{\boldsymbol{\theta}}_{(1)}, \widehat{\boldsymbol{\theta}}_{(2)}]$, consider the testing problem between the $\frac{p^{\frac{7}{8}k - 1}}{k!} 2^{\lfloor k/8 \rfloor + 1} + 1$ hypotheses $H_j: \boldsymbol{\theta}^*_{(2)}=\boldsymbol{\theta}_{(2)}^j \in \mathcal{Z}^\star$, using the minimum distance decoding strategy. By construction, if $k > 16$, we have:
\begin{align}
\label{eq:ar_sup_lowerbd}
\sup_{\mathcal{Z}^\star}\, \mathbb{P}\left[\left\|\widehat{\boldsymbol{\theta}}-{\boldsymbol{\theta}}^* \right \|_2 \geq \frac{h}{2} \right]= \sup_{j}\, \mathbb{P}\left[\widehat{\boldsymbol{\theta}} \neq {\boldsymbol{\theta}^j}\Big|H_j \right],
\end{align}
given that $\|\boldsymbol{\theta}_{(1)}^i\|_2 = e^{-b}$ or $0$. Let $f_{\boldsymbol{\theta}^j}$ denote joint probability distribution of $\{x_k\}_{k=1}^n$ conditioned on $\{x_k\}_{k=-p+1}^0$ {and $\{y_k\}_{k=-p+1}^n$} under the hypothesis $H_j$. Fano's inequality given in Lemma \ref{lem:ar_fano_ineq} provides the following lower bound on the probability in \prettyref{eq:ar_sup_lowerbd}:
\begin{equation}
\sup_j \mathbb{P}\left[\widehat{\boldsymbol{\theta}}\neq {\boldsymbol{\theta}^j}|H_j\right] \geq 1 - \frac{\xi+\log2}{\log_2 (|\mathcal{Z}^\star|-1)},
\end{equation}
where $\xi$ is an upper bound on the Kullback-Leibler divergence between $f_{\boldsymbol{\theta}^i}$ and $f_{\boldsymbol{\theta}^j}$ for any $i \neq j$, i.e., $\mathcal{D}_{\sf KL}\left(f_{\boldsymbol{\theta}^i}\Big\|f_{\boldsymbol{\theta}^j}\right) \le \xi, \forall i \neq j$. Given Gaussian innovations in the BVAR process, for $i \neq j$, we have
\begin{align}
\label{eq:ar_kl_bound}
\notag
\mathcal{D}_{\sf KL}\left(f_{\boldsymbol{\theta}^i}\Big\|f_{\boldsymbol{\theta}^j}\right) &\leq \sup_{i\neq j} \mathbb{E}\left[\log \frac{f_{\boldsymbol{\theta}^i}}{f_{\boldsymbol{\theta}^j}} \Big|H_i\right]\\
\notag & = \sup_{i \neq j} \mathbb{E}\left[-\frac{1}{2(\Sigma_{\epsilon})_{1,1}} \left( \left\| \mathbf{x} - \mathbf{X} \boldsymbol{\theta}^i \right \|^2 \right. \right. \\
\notag & \qquad \qquad \qquad \qquad \quad \, \left. \left. -  \left\| \mathbf{x} - \mathbf{X} \boldsymbol{\theta}^j \right \|^2\right) \Big| H_i\right]\\
\notag & \leq \sup_{i \neq j} \frac{1}{2(\Sigma_{\epsilon})_{1,1}} \mathbb{E}\left[ \left\| \mathbf{X} \left(\boldsymbol{\theta}^i-\boldsymbol{\theta}^j\right)\right\|^2 \Big| H_i \right]\\
\notag & = \frac{n}{2(\Sigma_{\epsilon})_{1,1}}\sup_{i \neq j} \left(\boldsymbol{\theta}^i-\boldsymbol{\theta}^j\right)^\top \mathbf{C} \left(\boldsymbol{\theta}^i-\boldsymbol{\theta}^j\right)\\
\notag & \leq \frac{n {\Lambda_{\max}\left(\mathbf{C}\right)}}{2(\Sigma_{\epsilon})_{1,1}}\sup_{i \neq j} \left \|\boldsymbol{\theta}^i-\boldsymbol{\theta}^j\right\|_2^2\\
& \leq \frac{2n (k+1) {\Lambda_{\max}\left(\mathbf{C}\right)}e^{-2b}}{(\Sigma_{\epsilon})_{1,1}} =: \xi.
\end{align}
Using \prettyref{lem:ar_fano_ineq}, Eqs. \prettyref{eq:ar_theta_lowerbd}, \prettyref{eq:ar_sup_lowerbd} and \prettyref{eq:ar_kl_bound} yield:
\begin{equation}
\sup_{\mathcal{Z}^\star}\mathbb{P}\left[\left\|\widehat{\boldsymbol{\theta}}\!-\!{\boldsymbol{\theta}}^* \right \|_2 \! \geq \! \frac{h}{2} \right] \! \geq \! 1-\frac{2\left( \frac{2n (k+1) {\Lambda_{\max}\left(\mathbf{C}\right)}e^{-2b}}{(\Sigma_{\epsilon})_{1,1}}\!+\!\log 2\right)}{(k+1) \log p}.
\end{equation}
for $p$ large enough so that $\log p \ge \frac{8 \log k - \log 2}{3 - \frac{12}{k}}$. Assuming $(k+1) \log p \ge 8 \log 2$ and choosing $b \ge \frac{1}{2} \log \left(\frac{16 {\Lambda_{\max}\left(\mathbf{C}\right)}}{(\Sigma_{\epsilon})_{1,1} } \right) + \frac{1}{2}\log \left(n/\log p\right)$,  we get:
\begin{align}
\sup_{\mathcal{Z}^\star}\mathbb{P}\left[\left\|\widehat{\boldsymbol{\theta}}-{\boldsymbol{\theta}}^* \right \|_2 \geq \frac{h}{2} \right] \geq \frac{1}{2},
\end{align}
Thus, there exist a choice in $\mathcal{Z}^*$ for which the error of an arbitrary estimator is larger than $h/2$ with probability at least $1/2$. By this choice of $\boldsymbol{\theta}^*$, the condition on $b$ gives $\|\boldsymbol{\theta}_{(2)}^* \|_2 \le \|\boldsymbol{\theta}^* \|_2 \le  \widetilde{\mathscr{B}} {k \log p}/{n}$, where $\widetilde{\mathscr{B}} = \frac{(\Sigma_{\epsilon})_{1,1}}{16 {\Lambda_{\max}\left(\mathbf{C}\right)}}$.

Next, under the alternative hypothesis, for a given threshold $\mathscr{t}$, the null hypothesis is not rejected if
\begin{equation}\label{eq:rejection}
\frac{\ell\Big(\widehat{\widetilde{\boldsymbol{\theta}}}_{(1)},\boldsymbol{0}\Big)}{\ell\Big(\widehat{\boldsymbol{\theta}}_{(1)}, \widehat{\boldsymbol{\theta}}_{(2)}\Big)} \le 1+ \mathscr{t},
\end{equation}
Using \prettyref{lem:normal_conc}, \prettyref{lem:reducedmodel-fullmodel_error}, and  \prettyref{prop:reduced model deviation}, the numerator on the left hand side of \prettyref{eq:rejection} can be upper bounded by $(\boldsymbol{\Sigma}_{\epsilon})_{1,1} + D + \Delta_D + \Delta_R + \Delta_N$. Thus, \prettyref{eq:rejection} holds if
\begin{equation}\label{eq:lower_fano}
\ell\Big(\widehat{\boldsymbol{\theta}}_{(1)}, \widehat{\boldsymbol{\theta}}_{(2)}\Big) \ge \frac{(\boldsymbol{\Sigma}_{\epsilon})_{1,1} + D + \Delta_D + \Delta_R + \Delta_N}{(1+\mathscr{t})}.
\end{equation}
For large enough $n$ and $p$, the numerator in \prettyref{eq:lower_fano} can be further upper bounded by $5 (\boldsymbol{\Sigma})_{1,1}$ with high probability. Next, we need to lower bound the left hand side of \prettyref{eq:lower_fano}. We  have $\ell\Big(\widehat{\boldsymbol{\theta}}_{(1)}, \widehat{\boldsymbol{\theta}}_{(2)}\Big) = \frac{1}{n} \|\mathbf{X}(\widehat{\boldsymbol{\theta}} - \boldsymbol{\theta}^*) \|^2_2 \ge \frac{\alpha(m+1)}{2m} \| \widehat{\boldsymbol{\theta}} - \boldsymbol{\theta}^*\|_2^2$, given the RE condition and the assumption on $\tau$ in \prettyref{prop:full model deviation}. Under the condition that $\|\widehat{\boldsymbol{\theta}}-{\boldsymbol{\theta}}^* \|_2 \geq \frac{h}{2}$, which holds with probability at least $1/2$, if 
\begin{equation}\label{eq:condition_h}
\frac{\alpha(m+1)}{8m} h^2 > \frac{5(\boldsymbol{\Sigma}_{\epsilon})_{1,1}}{(1+\mathscr{t})},
\end{equation}
then the inequality \prettyref{eq:lower_fano} holds, which implies that the null hypothesis is not rejected. Replacing $h$ from \prettyref{eq:ar_theta_lowerbd} into the inequality \prettyref{eq:condition_h} gives the condition $n < \widetilde{\mathscr{D}}' k \log p$, where $\widetilde{\mathscr{D}}' : = \frac{\alpha (m+1)(1 + \mathscr{t})}{{{10240}m\Lambda_{\max}\left(\mathbf{C}\right)}}$.
\end{proof}

It is noteworthy that while the result of Theorem \ref{thm:ness cond} imposes necessary conditions for reliable GC detection via LGC, it is a weak converse result as it does not quantify necessary conditions that hold for \emph{all} test statistics.

\subsection{False Positive Error Analysis}

Theorem \ref{thm:main-theorem} provides sufficient conditions that implicitly guarantee high sensitivity and specificity in detecting GC influences using the LGC statistic. To make these implications more explicit, by slightly weakening the sufficient condition on the sample size, $n\geq \mathscr{D}{''}k \log(2p)$ in \prettyref{thm:main-theorem}, we arrive at the following corollary to Theorem \ref{thm:main-theorem} that upper bounds the false positive error probability:
\setcounter{thm}{0}
\begin{cor}[False Positive Error Probability] \label{cor:false detection probability}
Suppose that assumptions \ref{assum:key} and \ref{assum:sparsity} as well as conditions \ref{cond:RE} and \ref{cond:DB} in the proof of \prettyref{thm:main-theorem} hold.
Then, for some arbitrary \(t_0 > 0\), thresholding the proposed LGC statistic $\mathcal{T}_{y \mapsto x}$ at a level \(\mathscr{t} >0\) for rejecting the null hypothesis results in a false positive error probability of at most $2\exp \left ( \displaystyle -{n}\Big/{8\left(1 + \gamma t_0 \sqrt{{\log(2p)}/{n}}\right)^2} \right)$
with \(\gamma:={(\mathscr{t}+2)}/{\mathscr{t}}\), if 
\begin{align*}
n \geq 2 \max\{\widetilde{\mathscr{D}}^2k^2/t_0^2, 2\widetilde{\mathscr{D}}\gamma k\} \log(2p),
\end{align*} 
for some constant $\widetilde{\mathscr{D}}$ that is explicitly given in proof. 
\end{cor}

\begin{proof} First, we note that under \prettyref{cond:RE} and \prettyref{cond:DB} there exist some real numbers $s,t>0$ such that
\begin{subequations}
\begin{align} 
\Big|\ell\Big(\boldsymbol{\theta}^*_{(1)}, \boldsymbol{\theta}^*_{(2)}\Big) - (\boldsymbol{\Sigma}_\epsilon)_{1,1}\Big| &\leq (\boldsymbol{\Sigma}_\epsilon)_{1,1}/{s}, \\
\Delta_{F}&\leq (\boldsymbol{\Sigma}_\epsilon)_{1,1} t/{s}.
\end{align}
\end{subequations}
This allows us to set a problem independent threshold, since the upper bound on $\mathcal{T}_{y \mapsto x}$ under the null hypothesis given in \prettyref{eq:ub null} simplifies to:
\begin{align}
\mathcal{T}_{y \mapsto x} \leq\frac{{2t}/{s}}{1-{(1+t)}/{s}} \label{eq:ub-null(1)}.
\end{align}
Now, given any threshold \(\mathscr{t} > 0\), we can solve for $s$ in terms of $t$ and $\mathscr{t}$ as:
\begin{align}\label{eq:s-t}
s = 1 + \frac{2+\mathscr{t}}{\mathscr{t}}t = 1 + \gamma t.
\end{align}

\noindent To ensure $\Delta_{F}\leq (\boldsymbol{\Sigma}_\epsilon)_{1,1} t/{s}$, we need the following to hold:
\begin{align} 
\label{eq:cor_sampling}
\frac{(m+1)\alpha}{42 m}\frac{(\boldsymbol{\Sigma}_\epsilon)_{1,1}}{16\mathscr{A}^2}&\frac{n}{k\log(2p)} - \frac{1}{t} - \gamma \geq 0.
\end{align}
On the other hand, using the expression for $s$ from \prettyref{eq:s-t} and invoking \prettyref{lem:fullmodel_error} (Appendix \ref{sec:Concentration-Lemmas}) yield the following statement that can be used to bound the false positive error probability:
\begin{align}
\notag &\mathbbm{P}\left[\Big|\ell\Big(\boldsymbol{\theta}^*_{(1)}, \boldsymbol{\theta}^*_{(2)}\Big) - (\boldsymbol{\Sigma}_\epsilon)_{1,1}\Big|\geq\frac{(\boldsymbol{\Sigma}_\epsilon)_{1,1}}{s}\right]\\
& \qquad \qquad \qquad \qquad \qquad \leq 2\exp\left(-\frac{n}{8\left(1 + \gamma t\right)^2}\right).
\end{align}

\noindent With a choice of \(t = t_0\sqrt{\log(2p)/n}\) for any \(t_0>0\), applying \prettyref{lem:useful-lemma} (Appendix \ref{sec:Concentration-Lemmas}) on \prettyref{eq:cor_sampling} then gives the sampling requirement of
\begin{align*}
n \geq \left(\widetilde{\mathscr{D}}/t_0\right)^2 {k^2}\log(2p) + 2\widetilde{\mathscr{D}}\gamma k\log(2p),
\end{align*}  
and the false positive error probability given in the corollary, where
\begin{align*}
\widetilde{\mathscr{D}}=\frac{42 m}{(m+1)\alpha}\frac{16\mathscr{A}^2}{(\boldsymbol{\Sigma}_\epsilon)_{1,1}}.
\end{align*}
This concludes the proof of the corollary.

Note that by further choosing $t_0 \le \sqrt{\widetilde{\mathscr{D}}k/ 2\gamma}$, one can simplify the sampling requirement and the corresponding upper bound on false positive error probability to \(n \ge 2 \left({\widetilde{\mathscr{D}}}\big/{t_0}\right)^2 k^2 \log(2p)\) and 
\(
2\exp\left(-{n}/{\left(1 + \sqrt{8}\right)^2}\right),
\)
respectively. The later inequality follows from the fact \(n \ge 4\widetilde{\mathscr{D}}k\gamma \log(2p)\) which is also guaranteed by choice of \(t_{0}\).
\end{proof}

\subsection{Power Analysis}

Intuitively, detecting a GC effect arising from a small cross-regression coefficient $\boldsymbol{\theta}^*_{(2)}$ is challenging, and often requires a long observation horizon to be identified.
\prettyref{thm:main-theorem} quantifies this intuition via a lower bound on the norm of the cross-regression coefficients in terms of the spectral properties of the process {(via \(\mathscr{B}\)), sparsity $k$,} sample size $n$ and model order $p$.
In particular, as \(\Vert \boldsymbol{\theta}^*_{(2)}\Vert_2 \rightarrow 0\), a scaling of $n = \mathcal{O}(k \log (2p) / \Vert \boldsymbol{\theta}^*_{(2)}\Vert_2^2)$ maintains the {sensitivity/specificity} of \(\mathcal{T}_{y\mapsto x}\) with high probability. This lower bound on $\Vert \boldsymbol{\theta}^*_{(2)}\Vert_2$ exhibits the same scaling as that in the thresholding procedure of \cite{Basu2015NGC} (i.e., the scaling of the LASSO estimation error), as well as the classical scaling of \cite{Wald1943,Davidson1970} (up to logarithmic factors), and we thus believe is not significantly improvable.

In Corollary \ref{cor:false detection probability}, we formalized the foregoing argument in terms of the test specificity. In order to similarly quantify the test sensitivity, we consider the effect size $\mathscr{r} := \frac{1}{k} \| {\boldsymbol{\theta}^*_{(2)}}\|^2$ pertaining to the cross-regression coefficients. The power analysis for detecting a GC effect using LGC is given in the following corollary to Theorem \ref{thm:main-theorem}:

\begin{cor}[Power Analysis] \label{cor:power}
Suppose that assumptions \ref{assum:key} and \ref{assum:sparsity} as well as conditions \ref{cond:RE} and \ref{cond:DB} in the proof of \prettyref{thm:main-theorem} hold.
Then, for an effect size $\mathscr{r} := \frac{1}{k} \| {\boldsymbol{\theta}^*_{(2)}}\|^2$ and an arbitrary constant $0 < \psi < 1$, thresholding the proposed LGC statistic $\mathcal{T}_{y \mapsto x}$ at a level \(\mathscr{t} >0\) gives a test power of at least $1 - 2\exp\left(-\frac{n}{128}\right) - 2\exp(-cn\min\{\zeta^{-2}\log(2p)/n,1\})$, if $\mathscr{r} \ge \left(\frac{\mathscr{r}_0}{\psi^2} + \frac{\mathscr{r}_1 + \mathscr{r}_2 \mathscr{t}}{\psi}\right) \frac{\log(2p)}{n}$ and 
\begin{align*}
n\geq \max \left \{ \frac{\mathscr{C}''}{(1-\psi)^2}, \mathscr{D}'k \right \} \log(2p),
\end{align*} 
for the same constant $\mathscr{D}'$ in Theorem \ref{thm:main-theorem} and some constants $\mathscr{r}_0$, $\mathscr{r}_1$, $\mathscr{r}_2$, and $\mathscr{C}''$ that are explicitly given in proof.  
\end{cor}

\begin{proof} 
Starting from \prettyref{eq:lb non null}, we seek sufficient conditions to ensure
\begin{align}\label{eq:t_condition_1}
\mathscr{t} < \frac{D-\left(\Delta_{D}+\Delta_{R}+\Delta_{F}\right)}{(\boldsymbol{\Sigma}_\epsilon)_{1,1}+\Delta_{N}+\Delta_{F}}
\end{align}
with high probability for a given threshold $\mathscr{t}$. Recall that:
\begin{align}
\nonumber \Delta_{F} &:= \frac{24}{m+1}\frac{k\lambda_{n}^{2}}{\alpha/m},\\
\nonumber \Delta_{R} &:= 20\frac{k\lambda_{n}^{2}}{\alpha/m} +\left(8\sqrt{2m}+18\right)\lambda_{n}\left\Vert \widetilde{\boldsymbol{\theta}}^*_{(1)J^{c}}\right\Vert _{1},\\
\nonumber D &:={\boldsymbol{\theta}^*_{(2)}}^{\top}(\mathbf{C}_{22}-\mathbf{C}_{21}\mathbf{C}_{11}^{-1}\mathbf{C}_{12}){\boldsymbol{\theta}^*_{(2)}},
\end{align}
and
\begin{align*}
\Delta_D &:= \mathbbm Q(\boldsymbol{\theta}^{*},\boldsymbol{\Sigma}_\epsilon)\sqrt{\frac{\log(2p)}{n}} \left\Vert\left[-\mathbf{C}_{11}^{-1}\mathbf{C}_{12}{\boldsymbol{\theta}^*_{(2)}};{\boldsymbol{\theta}^*_{(2)}}\right]\right\Vert_1\\
&\quad + \frac{\alpha}{27}\left\Vert\left[-\mathbf{C}_{11}^{-1}\mathbf{C}_{12}{\boldsymbol{\theta}^*_{(2)}};{\boldsymbol{\theta}^*_{(2)}}\right]\right\Vert_2^{2},
\end{align*}
with a choice of $\lambda_{n}=4\mathscr{A}\sqrt{{\log(2p)}/{n}}$. By rearranging \prettyref{eq:t_condition_1}, we equivalently seek conditions to ensure:
\begin{equation}\label{eq:t_condition_2}
D > \Delta_D + \Delta_R + (1 + \mathscr{t})\Delta_F + \mathscr{t} \Delta_N + \mathscr{t} (\boldsymbol{\Sigma}_\epsilon)_{1,1}.
\end{equation}
As in the proof of Theorem \ref{thm:main-theorem}, we have:
\begin{align*}
\Delta_{D}+\Delta_{R} \leq& \, \mathscr{a}\sqrt{\frac{\log(2p)}{n}}\left\Vert {\boldsymbol{\theta}^*_{(2)}}\right\Vert_2^{2}+\mathscr{b}\sqrt{\frac{k\log(2p)}{n}}\left\Vert {\boldsymbol{\theta}^*_{(2)}}\right\Vert_2\\
& +\mathscr{c}'\frac{k\log(2p)}{n},
\end{align*}
where
\begin{align*}
\mathscr{a} &= \frac{\alpha}{27}\left(\left\Vert \mathbf{C}_{11}^{-1}\mathbf{C}_{12}\right\Vert_2^2 + 1\right),\\
\mathscr{b} &= \mathscr{A}\left[\left(32\sqrt{2m} + 73\right)\left\Vert \mathbf{C}_{11}^{-1}\mathbf{C}_{12}\right\Vert + 1\right],\\
\mathscr{c}' &=\frac{320m\mathscr{A}^{2}}{\alpha},
\end{align*}
and $D\geq\widetilde{\Lambda}_{\min}\big\Vert {\boldsymbol{\theta}^*_{(2)}}\big\Vert_2^{2}$, with $\widetilde{\Lambda}_{\min}:=\Lambda_{\min}\left(\mathbf{C}_{22}-\mathbf{C}_{21}\mathbf{C}_{11}^{-1}\mathbf{C}_{12}\right)$. Choosing $\Delta_N = {(\boldsymbol{\Sigma}_\epsilon)_{1,1}}/{4}$ and assuming that the upper bound \prettyref{eq:delta F} on $\lambda_n$ is met with equality, we have $\Delta_{F} = {(\boldsymbol{\Sigma}_\epsilon)_{1,1}}/{4}$ and hence \prettyref{eq:t_condition_2} holds if:
\begin{align}\label{eq:t_condition_3}
\nonumber \widetilde{\Lambda}_{\min}\left\Vert {\boldsymbol{\theta}^*_{(2)}}\right\Vert_2^{2} \geq& \, \mathscr{a}\sqrt{\frac{\log(2p)}{n}}\left\Vert {\boldsymbol{\theta}^*_{(2)}}\right\Vert_2^{2}+\mathscr{b}\sqrt{\frac{k\log(2p)}{n}}\left\Vert {\boldsymbol{\theta}^*_{(2)}}\right\Vert_2\\
& +\mathscr{c}''\frac{k\log(2p)}{n},
\end{align}
where $\mathscr{c}'' := \frac{16m\mathscr{A}^{2}}{\alpha}\left(\frac{24(6 \mathscr{t} + 1)}{m+1}+20\right) $. By an application of \prettyref{lem:useful-lemma}, the inequality in \prettyref{eq:t_condition_3} holds if
\begin{align}\label{eq:inequality_power}
\nonumber \left\Vert {\boldsymbol{\theta}^*_{(2)}}\right\Vert^{2}_2 \geq& \!\left(\!\!\frac{\mathscr{b}^{2}}{\left(\widetilde{\Lambda}_{\min}\!-\!\mathscr{a}\sqrt{\frac{\log(2p)}{n}}\right) ^2}+\frac{2\mathscr{c}''}{\left(\widetilde{\Lambda}_{\min}\!-\!\mathscr{a}\sqrt{\frac{\log(2p)}{n}}\right)}\!\!\right)\\
& \times \frac{k\log(2p)}{n}.
\end{align}
Given that $\| {\boldsymbol{\theta}^*_{(2)}}\|_2^2 = k \mathscr{r}$ and using the condition $n\geq \frac{\mathscr{C}''}{(1-\psi)^2} \log(2p)$ with
$\mathscr{C}'' := {\mathscr{a}^2}/{\widetilde{\Lambda}_{\min}^2}$,
{the following bound on the effect size $\mathscr{r}$ ensures the inequality in \prettyref{eq:inequality_power}}:
\begin{align*}
\mathscr{r} \ge \left(\frac{\mathscr{r}_0}{\psi^2} + \frac{\mathscr{r}_1 + \mathscr{r}_2 \mathscr{t}}{\psi}\right) \frac{\log(2p)}{n},
\end{align*}
where
\begin{align*}
\mathscr{r}_0 &:= \frac{\mathscr{b}^2}{\widetilde{\Lambda}_{\min}^2},\\
\mathscr{r}_1 &:= \frac{32 m \mathscr{A}^2}{\alpha\widetilde{\Lambda}_{\min}}\left( \frac{24}{m+1} + 20\right),\\
\mathscr{r}_2 &:= \frac{4608 m \mathscr{A}^2}{\alpha (m+1)\widetilde{\Lambda}_{\min}}.
\end{align*}

Finally, assuming that Conditions \ref{cond:RE} and \ref{cond:DB} hold, the deviation bound for the full model with $\Delta_N = {(\boldsymbol{\Sigma}_\epsilon)_{1,1}}/{4}$ holds with probability at least 
\(1-2\exp\left(-{n}/{128}\right)\) and the deviation bound under the reduced model holds with probability at least
\(1-2\exp(-cn\min\{\zeta^{-2}\log(2p)/n,1\})\), with the constant $\zeta$ in \prettyref{cond:DB}. Combining the two, the test power is at least
\begin{align*}
1 &- 2\exp\left(-\frac{n}{128}\right) - 2\exp(-cn\min\{\zeta^{-2}\log(2p)/n,1\}),
\end{align*}
This concludes the proof of the corollary.
\end{proof}

\subsection{Data-driven Choice of the Threshold for Hypothesis Testing}\label{sec:datadriven}
While Theorem \ref{thm:main-theorem} establishes the existence of a threshold $\mathscr{t}$ that can separate the null and alternative hypotheses, a key practical factor in the utility of LGC is to be able to set this threshold in a data-driven fashion. A careful inspection of the proofs of Theorem \ref{thm:main-theorem} and Corollaries \ref{cor:false detection probability} and \ref{cor:power} indeed allows us to estimate these thresholds in practice. Here, we consider two practical methods to choose the test threshold in a data-driven fashion: 

\subsubsection*{Method 1}
As in most practical applications of LASSO, we assume that $n$ is sufficiently large and satisfies the sufficient lower bound of Corollary \ref{cor:false detection probability} for a nominal choice of the free parameter $t_0$. Then, by setting $t_0 = 1$, we have the following proposition for the first method to choose the test threshold and the corresponding type I and II error characterization:

\setcounter{thm}{0}
\begin{prop}
Suppose that assumptions \ref{assum:key} and \ref{assum:sparsity} as well as conditions \ref{cond:RE} and \ref{cond:DB} in the proof of \prettyref{thm:main-theorem} hold. Given a target false positive probability $\pi_{\sf F}$, the threshold $\mathscr{t}$ can be chosen as:
\begin{align}\label{eq:t}
\mathscr{t} = 2 \bigg/ \left( \frac{n}{\sqrt{8 \log \left(\sfrac{2}{\pi_{\sf F}}\right) \log (2p)}} - \sqrt{\frac{n}{\log (2p)}} - 1\right).
\end{align}
Furthermore, if the effect size is large enough such that $\displaystyle \mathscr{r} \ge \left(\frac{\mathscr{r}_0}{\psi^2} + \frac{\mathscr{r}_1 + \mathscr{r}_2 \mathscr{t}}{\psi}\right) \frac{\log(2p)}{n}$, the type II error probability $\beta$ satisfies:
\begin{equation}
\beta \le 2\exp\left(-\frac{n}{128}\right) + 2\exp(-cn\min\{\zeta^{-2}\log(2p)/n,1\}).
\end{equation}
\end{prop}

\begin{proof}
This result is a specialization of Corollaries \ref{cor:false detection probability} and \ref{cor:power} to the specific choice of the threshold given in \prettyref{eq:t}.
\end{proof}

We used this method of setting the threshold in our simulation study in \prettyref{sec:Simulation-Studies} (Fig. \ref{fig:effect-of-varying-n-p}, dashed traces) as well as in our real data validation in \prettyref{sec:real data validation}. 

\subsubsection*{Method 2}
If the assumption of Method 1, that $n$ is large enough to satisfy the sufficient lower bound in \prettyref{cor:false detection probability}, is not valid, we need to choose the free parameter $t_0$ to be compatible with the lower bound on $n$ in Theorem \ref{thm:main-theorem}. To this end, we need to estimate three key parameters in a data-driven fashion: the sparsity level $k$, the noise variance $(\boldsymbol\Sigma_{\epsilon})_{1,1}$, and the curvature parameter $\alpha$ in the RE condition. We hereafter assume that $\tau$ in the RE condition is small enough so that $m = 2$. We also assume that $\lambda_n$ is empirically chosen via cross-validation. Also, let $\mathscr{M}(\mathbf{X})$ denote the mutual coherence of the matrix $\mathbf{X}$. We have the following proposition that provides the second method to choose the threshold $\mathscr{t}$ and the corresponding type I and II error analysis:

\begin{prop}
Suppose that assumptions \ref{assum:key} and \ref{assum:sparsity} hold and the sufficient conditions of Theorem \ref{thm:main-theorem} are met, but with $\mathscr{B}$ replaced by another constant $\mathscr{B}'$ (given in the proof). Given a target false positive probability $\pi_{\sf F}$, let $\pi_0$ be such that $\pi_{\sf F} \le \pi_0 + K'_1\exp\left(-n\bar{c}'\right) + \frac{K'_2}{p^{\bar{d}'}}$, for some known constants $K'_1, \bar{c}', K'_2$, and $\bar{d}'$. Then, the threshold $\mathscr{t}$ can be chosen as:
\begin{equation}
\mathscr{t} = 2 \bigg/ \left( \frac{1}{t_0} \left (\frac{n}{\sqrt{8 \log \left(\sfrac{2}{\pi_{0}}\right) \log (2p)}} - \sqrt{\frac{n}{\log (2p)}}\right) - 1\right),
\end{equation}
where 
\begin{equation}
t_0 = \frac{28 \widehat{k} \lambda_n^2}{\kappa\widehat{\alpha}(\widehat{\boldsymbol{\Sigma}}_\epsilon)_{1,1}}   \sqrt{\frac{\log(2p)}{8 \log \left(\sfrac{2}{\pi_{0}} \right)}},
\end{equation}
with
\begin{itemize} 
\item $\widehat{k} := \left\Vert \widehat{\boldsymbol{\theta}}^\delta \right\Vert_0$, where $\widehat{\boldsymbol{\theta}}^\delta$ is the estimate $\widehat{\boldsymbol{\theta}}$ thresholded at a level $\delta = \mathcal{O}(\lambda_n/\alpha)$,
\item $(\widehat{\boldsymbol{\Sigma}}_{\epsilon})_{1,1} = \frac{1}{n}\left\Vert\mathbf{x} - \mathbf{X}\widehat{\boldsymbol{\theta}}\right\Vert^2$, 
\item $\displaystyle \widehat{\alpha} := \frac{4}{3}\min_{0\leq i \leq2p} \left(\frac{1}{n}\mathbf{X}^\top\mathbf{X}\right)_{i,i} \Big(1 -\mathscr{M}(\mathbf{X}) \textstyle{\frac{\widehat{k}}{2\kappa}} \Big)$, where $\mathscr{M}(\mathbf{X})$ is the mutual coherence of $\mathbf{X}$,
\end{itemize}
and $\kappa$ is a constant given in the proof. Furthermore, if the effect size is large enough such that it satisfies the condition in Corollary \ref{cor:power}, the type II error probability $\beta$ satisfies:
\begin{equation}
\beta \le K''_1\exp\left(-n\bar{c}''\right) + \frac{K''_2}{p^{\bar{d}''}},
\end{equation}
where $K''_1, \bar{c}'', K''_2$, and $\bar{d}''$ are constants explicitly given in the proof.
\end{prop}

\begin{proof}
Given a choice of $t_0$, choosing the threshold $\mathscr{t}$ as:
\begin{equation}\label{eq:thresh}
\mathscr{t} = 2 \bigg/ \left( \frac{1}{t_0} \left (\frac{n}{\sqrt{8 \log \left(\sfrac{2}{\pi_{0}}\right) \log (2p)}} - \sqrt{\frac{n}{\log (2p)}}\right) - 1\right),
\end{equation}
ensures that the false positive error expression in Corollary \ref{cor:false detection probability} is bounded by $\pi_0$. On the other hand, for \prettyref{eq:cor_sampling} to hold, we need to have
\begin{equation}
\gamma \widetilde{\mathscr{D}} k \frac{\log (2p)}{n} \leq 1 - \frac{\widetilde{\mathscr{D}}k}{t_0} \sqrt{\frac{ \log (2p)}{n}}.
\end{equation}
We choose $t_0$ to satisfy the foregoing bound with equality, which along with the threshold in \prettyref{eq:thresh} leads to:
\begin{equation}
t_0 = k \widetilde{\mathscr{D}} \sqrt{\frac{\log(2p)}{8 \log \left(\sfrac{2}{\pi_0} \right)}} = k \frac{42 m}{(m+1)\alpha}\frac{\lambda_n^2}{(\boldsymbol{\Sigma}_\epsilon)_{1,1}}  \sqrt{\frac{\log(2p)}{8 \log \left(\sfrac{2}{\pi_{0}} \right)}}.
\end{equation}
For this choice of $t_0$, one needs to empirically estimate the sparsity level $k$, the noise variance $(\boldsymbol\Sigma_{\epsilon})_{1,1}$, and the curvature parameter $\alpha$ in the RE condition. For $k$, we use the estimator
$\widehat{k} = \left\Vert \widehat{\boldsymbol{\theta}}^\delta \right\Vert_0$,
where $\widehat{\boldsymbol{\theta}}^\delta$ is the estimate $\widehat{\boldsymbol{\theta}}$ thresholded at a level $\delta > 0$. It can be shown that if $\delta$ is chosen as $\mathcal{O}(\lambda_n/\alpha)$, then $\widehat{k} \ge (1 - {\mathscr{E}}) k$ for some constant $\mathscr{E}$, given Conditions \ref{cond:RE} and \ref{cond:DB} \cite[Theorem 7.3]{van2011adaptive}.

Next, we estimate $(\boldsymbol{\Sigma}_\epsilon)_{1,1}$ as 
\begin{equation}
(\widehat{\boldsymbol{\Sigma}}_{\epsilon})_{1,1} = \frac{1}{n}\left\Vert\mathbf{x} - \mathbf{X}\widehat{\boldsymbol{\theta}}\right\Vert^2,
\end{equation}
which with probability at least $1 - 2 \exp(-n/32) - c_{1}\exp(-c_{2}n\min\{\zeta^{-2},1\}) - {\frac{d_1}{(2p)^{d_2}}}$, satisfies
\begin{equation}
| (\widehat{\boldsymbol{\Sigma}}_{\epsilon})_{1,1} - (\boldsymbol{\Sigma}_{\epsilon})_{1,1}| \le \frac{(\boldsymbol{\Sigma}_{\epsilon})_{1,1}}{2} + \Delta_F,
\end{equation}
where $\Delta_{F} = \frac{24}{m+1}\frac{k\lambda_{n}^{2}}{\alpha/m}$.

Estimating the curvature parameter $\alpha$, however, is not trivial. We thus replace it with a lower bound using the mutual coherence of the design matrix, $\mathscr{M}(\mathbf{X})$, following a result from \cite{bickelSimultaneousAnalysisLasso2009} that holds with high probability:
\begin{align}
{\alpha} &\ge \frac{2m}{m+1}\min_{0\leq i \leq2p} \left(\frac{1}{n}\mathbf{X}^\top\mathbf{X}\right)_{i,i} \Big(1 - \mathscr{M}(\mathbf{X}) k\Big).
\end{align}
Recalling the sufficient condition $\| {\boldsymbol{\theta}^*_{(2)}}\|_2^2 \geq \mathscr{B}{k\log(2p)}/{n}$ in the statement of Theorem \ref{thm:main-theorem}, where $\mathscr{B}$ linearly depends on $\frac{m}{(m+1)\alpha}$, we weaken the lower bound on $\| {\boldsymbol{\theta}^*_{(2)}}\|_2^2$ as $\mathscr{B}'{k\log(2p)}/{n}$, where $\mathscr{B}'$ is the same as $\mathscr{B}$ with $\alpha$ replaced by the foregoing lower bound. The estimate $\widehat{\alpha}$ is obtained by using the lower bound with an estimate of $\frac{\widehat{k}}{1-\mathscr{E}}$ for $k$. Note that this lower bound on $\alpha$ may be too conservative in practice and may result in reducing the test power.

We next set $m=2$. For large enough $n$ and $p$, we also have $\Delta_F \le {(\boldsymbol{\Sigma}_{\epsilon})_{1,1}}/{2}$, so that the estimate of $t_0$, namely $\widehat{t}_0$, using the empirical estimates of $k$ and $(\boldsymbol{\Sigma}_\epsilon)_{1,1}$ satisfies
\begin{align}
\frac{\widehat{t}_0}{t_0} {\le} \frac{1 - \mathscr{E}}{2} := \kappa.
\end{align}
Thus, the false positive error can be bounded as
\begin{align}
\nonumber \pi_{\sf F} \le & \, \pi_0 + 2 \exp(-n/32) + c_{1}\exp(-c_{2}n\min\{\zeta^{-2},1\})\\
& + c_{3}\exp\left(-c_{4}n\min\{\zeta^{-2}\max\{D_0, D_0^{'}\},1\}\right) + {\frac{d_1}{(2p)^{d_2}}}.
\end{align}
Letting 
\begin{align*}
\bar{c}' &= \min\left\{\frac{1}{32}, c_2, c_4, c_2\zeta^{-2}, c_4\zeta^{-2}\max\{D_0, D_0^{'}\}\right\},\\
\bar{d}' &= d_2, \quad K'_1 = 2 + c_1 + c_3, \quad K'_2 = \frac{d_1}{2^{d_2}},
\end{align*}
provides the assumed bound on $\pi_F$.

To bound the type II error, we use the result of Corollary \ref{cor:power}, but need to ensure that the deviation condition for the reduced model also holds. Therefore, we have:
\begin{align}
\nonumber \beta \le& \, \, 2\exp\left(-{n}/{128}\right) + 2\exp\left(-{n}/{32}\right)\\
\nonumber &+ 2\exp(-cn\min\{\zeta^{-2}\log(2p)/n,1\})\\
\nonumber &+ c_{3}\exp\left(-c_{4}n\min\{\zeta^{-2}\max\{D_0, D_0^{'}\},1\}\right)\\
&+ {\frac{d_1}{(2p)^{d_2}}} + {\frac{d'_1}{(2p)^{d'_2}}}.
\end{align}
Letting 
\begin{align*}
\bar{c}'' & = \min\left\{\frac{1}{128}, c, c_4, c\zeta^{-2}, c_4\zeta^{-2}\max\{D_0, D_0^{'}\}\right\},\\
\bar{d}'' & = \min\{d_2,d'_2\}, \quad  K''_1 = 6 + c_3, \quad K''_2 = \frac{d_1}{2^{d_2}} + \frac{d'_1}{2^{d'_2}},
\end{align*}
establishes the bound on $\beta$. This concludes the proof of the proposition.
\end{proof}

While the choice of $t_0 = 1$ in Method 1 is quite convenient and results in a threshold independent of the details of the BVAR process, it may be too conservative in practice and may result in low test power if $n$ is not sufficiently large. Method 2, however, selects $t_0$ based on empirical estimates of process-dependent parameters, and is thus expected to provide a higher test power. We will further evaluate the type I and II error for these two methods for setting the test threshold using simulation studies in \prettyref{sec:threhsold_comparison} and compare their performance.
\subsection{Discussion of the Results}\label{sec:discussion}
To discuss the implications of these results, several remarks are in order: 
\begin{remark}\label{rem:1}
Unlike the conventional estimation error results of the LASSO that specify a lower bound on the regularization parameter $\lambda_n$, \prettyref{thm:main-theorem} prescribes a fixed choice of $\lambda_n$ for both the \emph{full} and \emph{reduced} estimation problems in \prettyref{eq:Lasso-full&reducedmoedel}.
This is due to an interesting phenomenon revealed by our analysis: while conventional analyses of LASSO focus on the estimation performance of a single model and thus provide a lower bound on $\lambda_n$, in our framework we have two competing models (i.e., \emph{full} and \emph{reduced}) which need to be distinguishable under the null and alternative hypotheses in order to reliably detect the GC influences.
The latter imposes an \emph{upper} bound on $\lambda_n$.
As such, there is a suitable interval for choosing $\lambda_n$ that results in both consistent estimation and discrimination of the two models.
We note that we have presented our theoretical results under the assumption that a single $\lambda_n$ within the aforementioned interval is used for both models, for the simplicity of analysis. The strategy of using the `same $\lambda_n$ for both full and reduced models' is also appealing from a practical perspective, since the empirical methods for choosing $\lambda_n$ generally involve solving the problems for multiple values of $\lambda_n$ to evaluate a selection criterion and to pick the \emph{optimal} $\lambda_n$. This process is typically computationally intensive and performing it only for the full model saves processing time. For example, the user may select the \emph{optimal} $\lambda_n$ via $k$-fold cross-validation for the full model which leads to smallest out of sample error across different folds (i.e., best predictive performance), and then use the resulting value of $\lambda_n$ for both the full and reduced models to estimate the VAR parameters and prediction error variance using the whole data, thus avoiding extra computational costs of cross-validation.
\end{remark}
\begin{remark}
\prettyref{cor:false detection probability} bounds the false positive error probability, i.e., type I error rate, for a simple thresholding {scheme} for detecting GC influences from \(\mathcal{T}_{y \mapsto x}\), under a slightly weakened sufficient condition on \(n\), i.e., $n =\mathcal{O} (k^{2}\log(2p))$ instead of $n = \mathcal{O}(k \log(2p))$.
This non-asymptotic result provides a principled guideline for choosing a threshold that controls the false positive error rate, as shown in \prettyref{sec:datadriven}.
As such, this {result} extends the conventional statistical testing framework based on the asymptotics of log-likelihood ratio statistic using OLS to the non-asymptotic setting using the LASSO.
This result can further be utilized in the assignment of \textit{p}-values (i.e., the false positive error probability when the observed LGC statistic is used as the threshold), as is commonly done in the conventional OLS-based setting based on the asymptotic $\chi^{2}$ distribution under the null hypothesis. 
\end{remark}
\begin{remark}[]
We have presented our results for a BVAR model in order to parallel the classical GC analysis. Our results can be extended to the general MVAR setting by using the \emph{conditional} notion of Geweke \cite{geweke1984lineardep} in a natural fashion, given that \prettyref{cond:RE} and \prettyref{cond:DB} readily generalize to this setting. To elaborate on this point, consider a case where the BVAR$(p)$ model in \prettyref{eq:true-MVAR(p)} is augmented by $(d-2)$ other time series to result in a general $d$-dimensional MVAR$(p)$ setting. 
Following the formulation of \prettyref{sec:unifying}, the conditional LGC statistic can be defined as:
\begin{align}
{\mathcal{T}}_{y\mapsto x}:=\frac{\ell\Big(\widehat{\widetilde{\boldsymbol{\theta}}}_{(1)},\boldsymbol{0}, \widehat{\widetilde{\boldsymbol{\theta}}}_{(3)}, \cdots, \widehat{\widetilde{\boldsymbol{\theta}}}_{(d)}\Big)}{\ell\Big(\widehat{\boldsymbol{\theta}}_{(1)}, \widehat{\boldsymbol{\theta}}_{(2)}, \widehat{\boldsymbol{\theta}}_{(3)}, \cdots, \widehat{\boldsymbol{\theta}}_{(d)}\Big)} - 1,
\label{eq:-cond-T-statistic}
\end{align}%
where $\widehat{\widetilde{\boldsymbol{\theta}}}_{(i)}$ and $\widehat{\boldsymbol{\theta}}_{(i)}$ denote the MVAR parameters from the $i^{\sf th}$ process to the first time series, under the \emph{reduced} and \emph{full} models, respectively. 
Then, using a similar procedure as in the proof of \prettyref{thm:main-theorem} and \prettyref{cor:false detection probability}, the results can be extended to this setting, by replacing the various occurrences of $2p$ by $dp$ and by adopting slightly different constants.
\end{remark}
\begin{remark}[]
The constants in the proofs of Theorems \ref{thm:main-theorem} and \ref{thm:ness cond} and Corollaries \ref{cor:false detection probability} and \ref{cor:power} solely depend on the joint spectrum of processes \(x_t, y_t\) as well as some absolute constants. As an illustrative example, by assuming $\boldsymbol{\Sigma}_{\epsilon}=0.01 \mathbf{I}$, $\mu_{\max}(\mathbf{A}) = 0.9$, $\mu_{\min}(\mathbf{A}) = \mu_{\min}(\breve{\mathbf{A}}) = 0.01$, $\widetilde{\Lambda}_{\min} = 0.7$, $\|\mathbf{C}_{11}^{-1}\mathbf{C}_{12}\|_2 = 0.5$, $\|[\mathbf{C}_{11}^{-1}\mathbf{C}_{12}; \mathbf{I}] {\boldsymbol{\theta}^*_{(2)}} \|_2 = 1.5$,
$m=8$, $d_0 =5.2\times 10^{-4}$, $d_0^' =4.2\times 10^{-4}$, $D_0 = 100$, $C_0 = 10^{-6}$, and $c = 0.02$, the key constants in \prettyref{thm:main-theorem} take the following numerical values:
$\mathscr{A} = 10^{-3}$, $\mathscr{B} = 5.136$, $\mathscr{C}'=7.41\times10^{-4}$, $\mathscr{D}=24.38$, $K_1 = 6$, $K_2=6$, $\bar{c}=2.06\times10^{-4}$ and $d = 1$. These translate to $\lambda_n = 10^{-3}\sqrt{{\log{(2p)}}/{n}}$, a requirement of $n > \max\{100, 24.38k\}\log(2p)$, local alternative hypotheses satisfying $\Vert \boldsymbol{\theta}^*_{(2)}\Vert_2^2 > 5.136k{\log{(2p)}}/{n}$, and failure probability $< 6\exp{\left(-2.06\times10^{-6}n\right)}+ {6}/{p}$. Similarly, for \prettyref{cor:false detection probability}, we get $\widetilde{\mathscr{D}} = 10.67$, translating to a sample size requirement of
$n > 2 \max\left\{ 28.46 k^2, 396k\right\}\log(2p)$
(with $t_0 = 2$ and $\mathscr{t}=0.114$). {The} potentially large numerical values of some of these constants suggest that the non-asymptotic advantage may come with large values of \(n\) and \(p\).
\end{remark}
\section{Application to Simulated and Experimentally-Recorded Data\label{sec:experimental validation}}
In this section, we examine our theoretical results through application to simulated and real data, and by comparing the performance of classical OLS-based GC and the proposed LGC statistic in detecting GC influences. We use the fast implementation in \cite{Golsdtein2009} to solve the LASSO problems. Unless otherwise stated, the regularization parameter $\lambda_{n}$ is chosen via five-fold cross-validation performed over the \emph{full} model, with the same $\lambda_{n}$ used for the \emph{reduced} model.

\subsection{Simulation Studies\label{sec:Simulation-Studies}}
We simulated three time series $x_{t},y_{t},z_{t}$ according to the sparse MVAR$(11)$ model:
\begin{align*}
x_{t} = &-0.67x_{t-1}+0.2x_{t-5}-0.1x_{t-11}+0.05z_{t-3}+\nu_{1,t},\\
y_{t} = &-0.62y_{t-1}+0.1y_{t-5}-0.2y_{t-11}-0.1x_{t-2}-0.1x_{t-3},\\
        &+ 0.5x_{t-11}-0.001z_{t-4}-0.004z_{t-5}+\sqrt{0.6}\nu_{2,t}\\
z_{t} = &-0.9025z_{t-2}+\nu_{3,t}.
\end{align*}
where $\nu_{i,t} \sim \mathcal{N}(0,1)$, i.i.d. for $i=1,2,3$. In this model, $x_t$ has a direct GC influence on $y_t$, but there is no GC influence from $y_t$ to $x_t$. The \emph{latent} process $z_t$, however, influences both $x_t$ and $y_t$ (\prettyref{fig:generative-model-and-lambdan-effect}(a)). As such, the correlated process noise components $\epsilon_t$ and $\epsilon'_t$ in \prettyref{eq:true-MVAR(p)} are modeled as $0.05z_{t-3}+\nu_{1,t}$ and $-0.001z_{t-4}-0.004z_{t-5}+\sqrt{0.6}\nu_{2,t}$, respectively. As shown in \prettyref{fig:generative-model-and-lambdan-effect}(b), removing $z_t$ from the analysis indeed induces a false (i.e., indirect) GC influence from $y_t$ to $x_t$. We performed two sets of numerical experiments to evaluate the effects of $\lambda_n$, $n$, and $p$ on the identification of GC influences between $x_t$ and $y_t$ based on $\mathcal{T}_{y \mapsto x}$ {and $\mathcal{T}_{x \mapsto y}$}:

\begin{figure*}[t!]
  \begin{center}
  \includegraphics[width=0.95\linewidth]{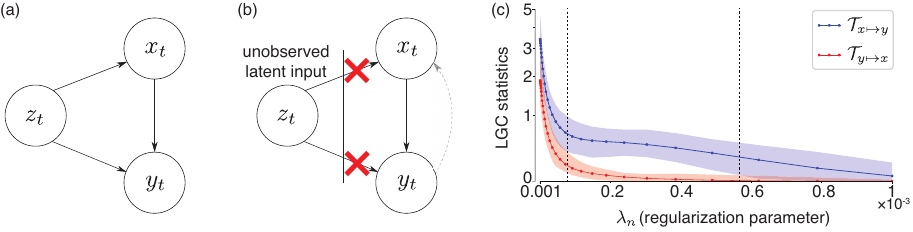}
  \end{center}
  \vspace{-3mm}
  \caption[LGC Simulation Results]{Simulation Results. (a) Ground truth GC pattern.  
  (b) Estimation setup, in which $z_t$ is latent and thus introduces a spurious GC link (dashed gray arrow) from $y_t$ to $x_t$. (c) Effect of $\lambda_{n}$ on the LGC statistics for $n=250, p=100$. The LGC statistics ${\mathcal{T}}_{y\mapsto x}$ (red) and ${\mathcal{T}}_{x\mapsto y}$ (blue) are separable for a suitable range of $\lambda_n$, marked by the dashed vertical lines (colored hulls show the range of LGC over \num{30} realizations).
  \label{fig:generative-model-and-lambdan-effect}}
\end{figure*}

\subsubsection[Evaluating the Effect of lambda_n]{Evaluating the Effect of $\lambda_n$}
\prettyref{fig:generative-model-and-lambdan-effect}(c) shows the LGC statistics $\mathcal{T}_{y \mapsto x}$ (red) and $\mathcal{T}_{x \mapsto y}$ (blue) for $n=250$ and $p=100$, obtained by varying $\lambda_n$ in the interval $[10^{-6}, 10^{-3}]$ uniformly in the log-scale. The dotted lines and colored hulls represent the average and range of the values, respectively, over $30$ realizations. As discussed in \emph{Remark \ref{rem:1}}, there is an evident range of $\lambda_{n}$ that provides a meaningful separation between ${\mathcal{T}}_{y\mapsto x}$ and ${\mathcal{T}}_{x\mapsto y}$, which is marked by the dashed vertical lines in {\prettyref{fig:generative-model-and-lambdan-effect}(c).

\subsubsection[Evaluating the Effect of the Sample Size n]{Evaluating the Effect of the Sample Size $n$}
We fixed a model order of $p=100$ and varied $n$ uniformly in the interval $[100,1000]$. \prettyref{fig:effect-of-varying-n-p}(a) and (b) show the resulting LGC statistics $\mathcal{T}_{y \mapsto x}$ and $\mathcal{T}_{x \mapsto y}$ corresponding to the LASSO and OLS ($\lambda_n = 0$), respectively. In \prettyref{fig:effect-of-varying-n-p}(a), we also plotted the threshold $\mathscr{t}$ corresponding to a false positive probability of $0.01$, according to Method 1 in \prettyref{sec:datadriven} (dashed line). As $n$ grows larger, the ranges of the two LGC statistics are saliently separated. The proposed thresholding rule of \prettyref{cor:false detection probability} is also able to correctly identify the true GC effects for $n \ge 250$. 

The OLS results shown in \prettyref{fig:effect-of-varying-n-p}(b), however, require much larger values of $n$ to be stable, whereas the LGC statistic provided by the LASSO (\prettyref{fig:effect-of-varying-n-p}(a)) are stable even for $n < 2p$. In addition, OLS requires $n\ge 400$ for the ranges of the GC measures to be distinguishable.
\subsubsection[Evaluating the Effect of the Model Order p]{Evaluating the Effect of the Model Order $p$} 
Finally, we fixed $n=300$ and varied $p$ in the interval $[10,300]$ uniformly in the log-scale. \prettyref{fig:effect-of-varying-n-p}(c) and (d) show the corresponding GC measures for the LASSO and OLS, respectively, along with the threshold $\mathscr{t}$ corresponding to a false positive probability of $0.01$. For $p \ll n$, the LASSO and OLS exhibit similar performance. But, the OLS-based GC measures become unstable for $p \approx n$, whereas those of the LASSO remain stable throughout. The LGC statistics also remain saliently separable over a wider range of $p$ for the LASSO, as compared to their OLS counterparts.
\subsubsection{Comparison with Existing LASSO-based Methods}
\label{sec:Simulation-Studies-comp}
We used the same setting as in \prettyref{sec:Simulation-Studies} to compare the performance of LGC with two existing LASSO-based methods, namely the confidence interval (CI)-based LASSO and truncating LASSO (TLASSO) GC detection. For the CI-based LASSO, we used de-biasing via node-wise regression \cite{van2014asymptotically} and the confidence interval construction technique of \cite{javanmard2018debiasing}: if the confidence interval of at least one of the cross-regression coefficients does not include zero (after Bonferroni correction), we identify it as a GC link. We used the Bonferroni correction for its simplicity and common usage in the applications of GC \cite{wu2021granger,shojaie2014inferring}. It is however possible to use other multiple comparison correction schemes such as the Benjamini-Hochberg procedure \cite{benjamini1995controlling}. The TLASSO, as a graphical LASSO-based method, uses the truncating LASSO penalty to automatically determine the order of the VAR models, and thus performs model simplification by reducing the number of covariates to control false discoveries \cite{Shojaie2010GGC}.  

\prettyref{fig:ROC} shows the receiver operating characteristic (ROC) plot for varying sample sizes, $n \in  \{ 200, 350, 500, 750,$ $1000, 1250, 1500\}$, computed from $200$ realization of the same process in \prettyref{sec:Simulation-Studies}. The marker sizes are proportional to $n$, for visual convenience. We fixed the BVAR order $p$ at $50$ for LGC and CI-based LASSO, while TLASSO automatically determined the BVAR order ($p \le 50$). For LGC, we considered a threshold for false positive error rate of $0.1$ based on Method 1 in \prettyref{sec:datadriven}, and accordingly set a confidence level of $90\%$ for the CI-based LASSO and a false negative rate of $0.1$ for TLASSO. The TLASSO is the most conservative of the three methods in terms of controlling the false positive error, which comes at the cost of low test power. On the other hand, the CI-based LASSO exhibits a high test power, but at cost of increasing false positive errors. LGC, however, adheres to the middle ground between these two extremes by striking a reasonable balance between true and false positive error rates. Note that this simulation study was intentionally designed to evaluate the performance of LGC in presence of weak GC influences. As expected, when the GC influence gets stronger, all three methods exhibit similar performance.

\begin{figure*}[t!]
\begin{center}
\includegraphics[width=0.85\linewidth]{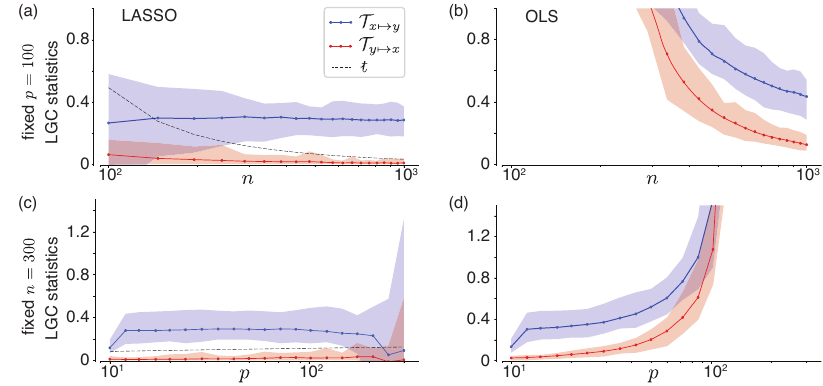}
\end{center}
\caption[LGC Simulation Results (continued)]{Simulation Results (continued). LASSO-based (left) and OLS-based {($\lambda_n = 0$, right)} {LGC} statistics ${\mathcal{T}}_{y\mapsto x}$ (red) and ${\mathcal{T}}_{x\mapsto y}$ (blue) obtained by varying $n$ for fixed $p=100$ (top panels (a) and (b)) and varying the model order $p$ for fixed $n=300$ (bottom panels (c) and (d)). The dashed lines in panels (a) and (c) show the threshold $\mathscr{t}$ at a false positive error level of $0.01$ (colored hulls show the range LGC over \num{30} realizations).
\label{fig:effect-of-varying-n-p} }
\end{figure*}

\begin{figure}[t!]
\begin{center}
\includegraphics[width=0.95\columnwidth]{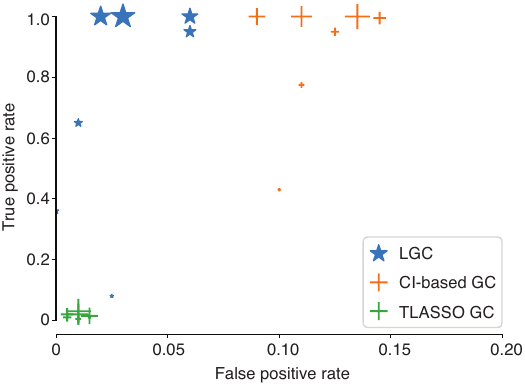}
\end{center}
\caption[]{{Comparing the performance of LGC, CI-based LASSO, and TLASSO for GC identification. True positive and false positive error rates are computed from $200$ realization of the BVAR process in \prettyref{sec:Simulation-Studies} with $p=50$ by varying $n \in \{ 200, 350, 500, 750, 1000, 1250, 1500 \}$. The marker sizes are proportional to $n$.
}\label{fig:ROC} }
\end{figure}

\vspace{-2mm}
\subsection{Comparing the Data-Driven Methods for Setting the Test Threshold}\label{sec:threhsold_comparison}

In the foregoing simulation studies, we used Method 1 in \prettyref{sec:datadriven} to set the test threshold due to its convenient form. In order to compare the two data-driven methods for setting the test threshold given in \prettyref{sec:datadriven}, we consider a slightly modified version of the simulation setting in \prettyref{sec:Simulation-Studies}:
\begin{align*}
x_{t} = &-0.1x_{t-1}+0.2x_{t-5}-0.1x_{t-11}+0.05z_{t-3}+\nu_{1,t},\\
y_{t}= &-0.1y_{t-1}-a\times0.1x_{t-2} + a\times0.5x_{t-11}-0.001z_{t-4}\\
\quad &-0.004z_{t-5}+\sqrt{0.6}\nu_{2,t},\\
z_{t} = &-0.9025z_{t-2}+\nu_{3,t}.
\end{align*}

In this modified setting, the sparsity level is reduced to $k=3$, and the leading VAR coefficients are decreased, so that the estimate of the curvature parameter $\alpha$ using the mutual coherence is reliable for small values of $n$ and $p$. The multiplier $0 \le a \le 1$ is chosen to control the effect size of the GC link. Similar to the previous case, we consider $p=50$, and use a range of $n \in \{ 100, 150, 200, 250, 300, 350\}$ to closely examine the role of the effect size in the performance of the two methods. To speed up simulations, we tuned $\lambda_n$ for $n=100$ via 2-fold cross-validation, and used the scaling $\sqrt{100/n}$ to obtain the subsequent values of $\lambda_n$.

\begin{figure*}[t!]
\begin{center}
\includegraphics[width=0.9\linewidth]{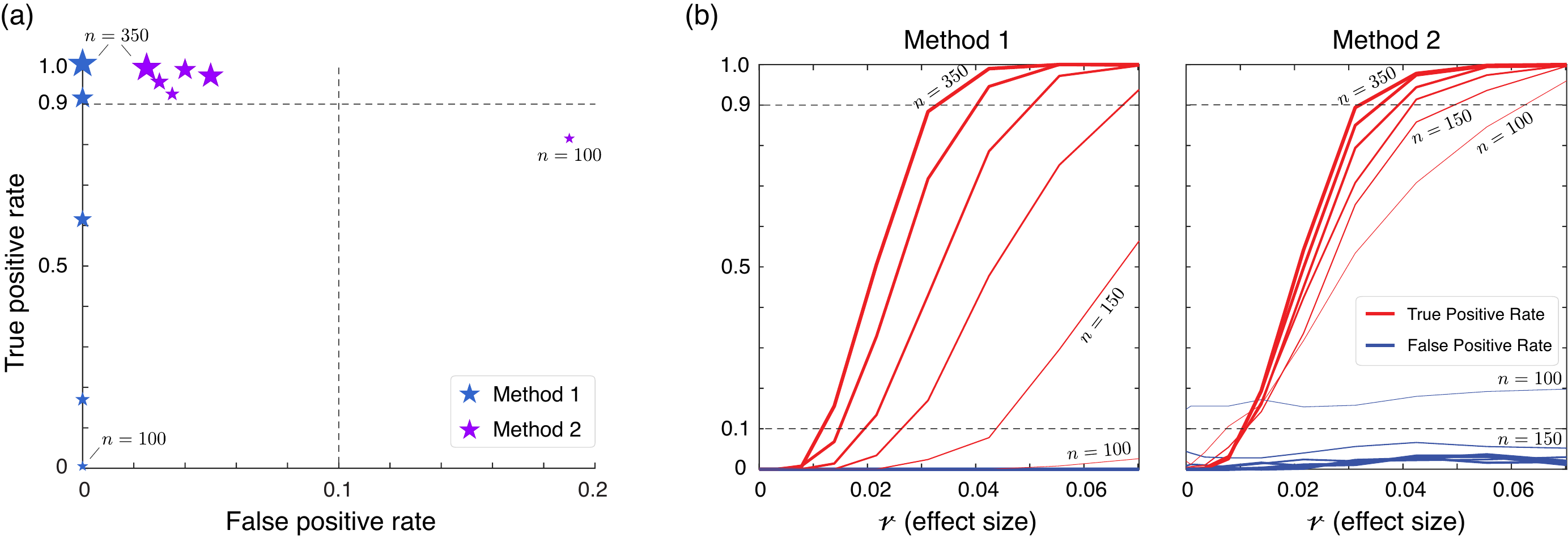}
\end{center}
\caption[]{Comparing the performance of the two methods of \prettyref{sec:datadriven} for selecting the threshold in a data-driven fashion. (a) True positive and false positive rates are computed from $200$ realization of the modified BVAR process with $a=0.75$, $p=50$, and varying $n \in \{ 100, 150, 200, 250, 300, 350 \}$, for a target false positive probability of $0.1$. The results for Methods 1 and 2 given in \prettyref{sec:datadriven} are shown by blue and purple stars, respectively. The marker sizes are proportional to $n$. (b) True positive (red) and false positive (blue) rates as a function of the effect size $\mathscr{r}$, by varying $a \in [0,1]$. Line widths are proportional to $n$.
\label{fig:th_comp} }
\end{figure*}%

Fig. \ref{fig:th_comp}(a) shows the ROC plot for the two methods for setting the test thresholds given in \prettyref{sec:datadriven}. The marker sizes are proportional to $n$. As it can be observed from the figure, Method 1 (blue stars) is expectedly more conservative and while it maintains negligible false positive rates, it requires larger values of $n$ for reliable detection of the GC link. Method 2 (purple stars), however, is more sensitive, but for $n \ge 150$ achieves both high test power (more than $0.9$) and a false positive error less than the target level of $0.1$.

Fig. \ref{fig:th_comp}(b) shows the performance of both methods with respect to the effect size. The line widths are proportional to $n$. Consistent with Fig. \ref{fig:th_comp}(a), Method 1 is more conservative and while consistently achieves a negligible false alarm rate, it requires a larger value of $n$ to detect GC links with small effect size, as indicated by the horizontal dashed line at $0.9$. Method 2, however, achieves a false alarm rate below the target level (horizontal dashed line at $0.1$) for $n \ge 150$, while reliably capturing GC links of small effect size. As expected, the performance of both methods become comparable as $n$ gets larger.

\subsection{Application to Experimentally-Recorded Neural Data from General Anesthesia\label{sec:real data validation}}
Finally, we present an application to simultaneous local field potentials (LFPs) and an ensemble of single-unit recordings from the temporal cortex in a human subject under Propofol-induced general anesthesia (Data from \cite{lewis2013local}). The LFP signal is the electrical field potential measured at the cortical surface and represents mesoscale dynamics of brain activity {with both cortical and sub-cortical (e.g., thalamic) origins}. Single-unit spike recordings, on the other hand, represent the neuronal scale cortical dynamics.

\begin{figure*}[t!]
\centering
  \includegraphics[width=0.95\linewidth]{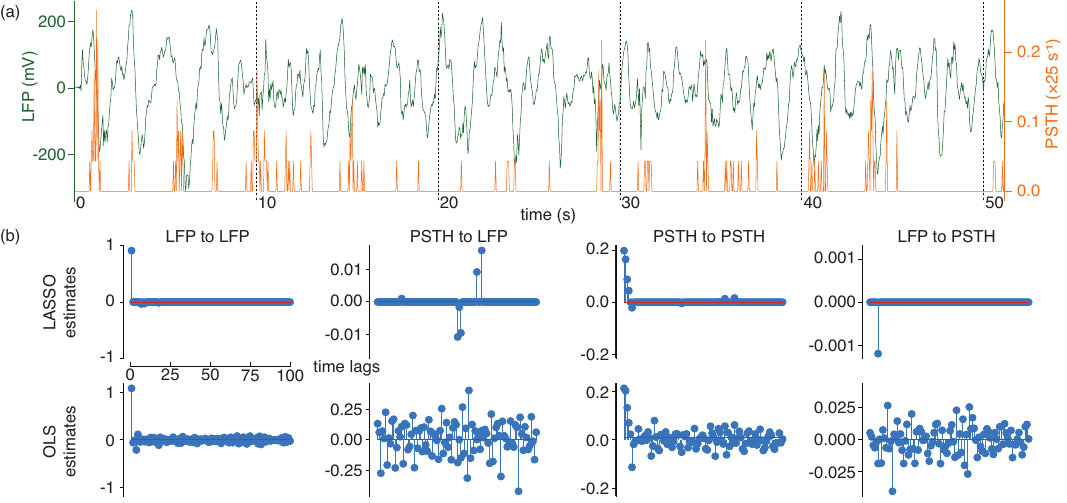}
  \caption[Analysis of neural data from general anesthesia]{Analysis of neural data from general anesthesia. (a) LFP (green) and PSTH (orange) traces for a time window of duration \SI{51.2}{\second}. (b) BVAR parameter estimates corresponding to the \emph{full} models for LASSO (top) and OLS (bottom).
  \label{fig:exp result} }
\end{figure*}

Brain states under anesthesia and sleep are associated with the emergence of periodic and profound suppression of neuronal spiking activity that is strongly phase-locked to the peaks of the LFP slow oscillations \cite{lewis2012rapid,lewis2013local,chauvette2011properties,Watson_2016}. 
{Specifically, by comparing the average LFP signals triggered at the trough of the slow oscillations under no-spike and many-spike conditions, \cite{lewis2012rapid} argues that neuronal spiking activity may have a GC influence in high-amplitude peaks of the slow oscillations manifested in the LFP. Here, we examine the role of neuronal spiking activity in mediating the LFP slow oscillations by assessing the GC influences between them.}

We use a time duration of \SI{51.2}{\second} during anesthesia, corresponding to $n = 1280$ samples (sampling frequency of \SI{25}{\hertz}). The ensemble spiking activity is represented by its peristimulus time histogram (PSTH) (i.e., ensemble average over \num{23} units). \prettyref{fig:exp result}(a) shows the LFP (green) and PSTH (orange) signals used in the analysis. We use a model order of $p=100$, corresponding to a history length of \SI{4}{\second}, to ensure that slow oscillations ($\sim$\SIrange{0.25}{0.5}{\hertz}) can be captured by the BVAR model. 
\prettyref{fig:exp result}(b) shows the estimated BVAR coefficients by the LASSO (top) and OLS (bottom). A visual comparison of the two sets of coefficients suggests that OLS has likely over-fitted the data. The corresponding LGC statistics $\mathcal{T}_{\sf LFP \mapsto PSTH}$ and $\mathcal{T}_{\sf PSTH \mapsto LFP}$ for both methods are reported in \prettyref{tab:lgc statistics}. The numbers in parentheses show the \textit{p}-values, i.e., the false positive error probability, when the observed GC statistics are used as thresholds. For the LGC statistics, the \textit{p}-values are computed via \prettyref{cor:false detection probability} (with a choice of \({t_0=0.25}\)), while in the conventional OLS setting we used the \(\chi^2\) distribution against the threshold \(\mathcal{F}_{x\mapsto y} = \log\left(1 + \mathcal{T}_{x\mapsto y}\right) \).

For a false positive error probability of $0.01$, the LGC-based test clearly detects the GC effect ${\sf PSTH \mapsto LFP}$ (boldface number) as significant, and discards the GC effect ${\sf LFP \mapsto PSTH}$. The conventional $\chi^2$ test applied to the classical GC statistics $\mathcal{F}_{\sf LFP \mapsto PSTH}$ and $\mathcal{F}_{\sf PSTH \mapsto LFP}$, however, fails to detect any GC influence, even at a significance level {as high as} $0.05$. The outcome of the LASSO-based LGC analysis is therefore consistent with the aforementioned hypothesis in \cite{lewis2012rapid} on the GC influence of spiking activity in mediating the LFP dynamics.

\begin{table}[!h]
\renewcommand{\arraystretch}{1.3}
\caption{Obtained LGC Statistics and \textit{p}-values}
\begin{center}
\begin{tabular}{|| c | c  c ||}
\hline
  & LASSO & OLS \\ 
  \hline \hline
 \(\mathcal{T}_{\sf LFP \mapsto PSTH}\) & \(0.0055~(0.1079)\) & \(0.1001~(0.1832)\) \\  
 \(\mathcal{T}_{\sf PSTH \mapsto LFP}\) & \({\bf 0.0096}~(\bf{0.0015})\) & \(0.1046~(0.1135)\)    \\
 \hline
\end{tabular}
\label{tab:lgc statistics}
\end{center}
\renewcommand{\arraystretch}{1}
\end{table}

\section{Concluding Remarks\label{sec:Conclusions}}
In this work, we proposed a GC statistic based on the LASSO parameter estimates, namely the LGC statistic, in order to identify GC influences in a canonical sparse BVAR model with correlated process noise. By analyzing the non-asymptotic properties of LGC statistic, we established that the well-known sufficient conditions for the consistency of LASSO also suffice for accurate identification of GC influences, if the strength of the GC effect is large enough. We also established the necessity of the constraint on the strength of the GC effect for reliable GC identification via LGC. We also analyzed the false positive error performance and test power of a simple thresholding rule for detecting GC influences and provided data-driven methods to set the test threshold in practice. We validated our theoretical claims through application to simulated and experimentally-recorded neural data from general anesthesia. In particular, we showed that the proposed LGC statistic is able to identify a GC effect from spiking activity to LFP slow oscillations under anesthesia, whereas the conventional OLS-based GC analysis does not detect this effect. Our contribution compared to existing literature is to provide a simple statistic inspired by the classical log-likelihood ratio statistic used for GC analysis, which can be directly computed from the LASSO estimates without the need to resort to de-biasing procedures or asymptotic results for testing. Future work includes extending our results to autoregressive generalized linear models with time-varying parameters and obtaining necessary conditions that hold for any test statistic.
\appendices
\section{Prediction Error Analysis of the Full and Reduced Models}
\label{sec:prediction-error-analysis}
In this section, we establish the deviation bounds of \prettyref{eq:full dev} and \prettyref{eq:reduced dev} under Conditions \ref{cond:RE} and \ref{cond:DB}.
Note that both the \emph{full} and \emph{reduced} models share the same RE condition (\prettyref{cond:RE}), since the \emph{reduced} model is nested within the \emph{full} model. However, the deviation conditions required by \prettyref{cond:DB} are different for the two models. For the \emph{full} model, we require:
\begin{align}\left\Vert \frac{1}{n}\mathbf{X}^{\top}\!({\mathbf{x\!}-\mathbf{X}\boldsymbol{\theta}^{*}})\right\Vert _{\infty}\leq \mathbbm Q(\boldsymbol{\theta}^{*},\boldsymbol{\Sigma}_{\epsilon})\sqrt{\frac{\log(2p)}{n}}, \label{eq:DB-1}
\end{align}
for some deterministic function $\mathbbm Q(\boldsymbol{\theta}^{*},\boldsymbol{\Sigma}_{\epsilon})$. In the \emph{reduced} model, however, we require
\begin{align}\left\Vert \frac{1}{n}\mathbf{X}_{(1)}^{\top}\left(\mathbf{x}-\mathbf{X}_{(1)}\widetilde{\boldsymbol{\theta}}^*_{(1)}\right)\right\Vert _{\infty}\leq \mathbbm Q'(\boldsymbol{\theta}^{*},\boldsymbol{\Sigma}_{\epsilon})\sqrt{\frac{\log(2p)}{n}} \label{eq:DB-2}\end{align}
for another deterministic function $\mathbbm Q'(\boldsymbol{\theta}^{*},\boldsymbol{\Sigma}_{\epsilon})$.
These conditions guarantee consistent and stable recovery of the autoregressive parameters in the LASSO problems of \prettyref{eq:Lasso-full&reducedmoedel}, and hence lead to suitable deviation results for both the \emph{full} and \emph{reduced} models.  

\cleartheorem{prop}
\newtheorem{prop}{\bf Proposition}[section]
\cleartheorem{lem}
\newtheorem{lem}{\bf Lemma}[section]
\cleartheorem{remark}
\newtheorem{remark}{\bf Remark}[section]

\begin{prop}[Deviation Result for the \emph{Full} Model] \label{prop:full model deviation}
Suppose $\widehat{\boldsymbol{\Sigma}}\sim\text{RE}\left(\alpha,\tau\right)$, with $\tau$ satisfying \({32k\tau}/{\alpha}={(m-1)}/{m}\)
for some $m>1$ and \((\mathbf{X}, \mathbf{x})\) satisfying the deviation bound \eqref{eq:DB-1}. 
Then for any \(\lambda_{n}\geq4 \mathbbm{Q}(\boldsymbol{\theta}^{*},\boldsymbol{\Sigma}_{\epsilon})\sqrt{{\log(2p)}/{n}},\) the solution to the \emph{full} model in \prettyref{eq:Lasso-full&reducedmoedel} satisfies: 
\begin{align}
\left\vert \ell\Big(\widehat{\boldsymbol{\theta}}_{(1)}, \widehat{\boldsymbol{\theta}}_{(2)}\Big)-\ell\Big(\boldsymbol{\theta}^{*}_{(1)}, \boldsymbol{\theta}^*_{(2)}\Big)\right\vert \leq \frac{24}{m+1}\frac{k\lambda_{n}^{2}}{\alpha/m}=:\Delta_{F}. \label{eq:full2}
\end{align}
\end{prop}
\begin{proof}
For the \emph{full} model, we have: 
\begin{align}
\notag \ell\Big(\widehat{\boldsymbol{\theta}}_{(1)}, \widehat{\boldsymbol{\theta}}_{(2)}\Big)-\ell\Big(\boldsymbol{\theta}^{*}_{(1)}, \boldsymbol{\theta}^*_{(2)}\Big)
&=\frac{1}{n}(\widehat{\boldsymbol{\theta}}-\boldsymbol{\theta}^{*})^{\top}\mathbf{X}^{\top}\mathbf{X}(\widehat{\boldsymbol{\theta}}-\boldsymbol{\theta}^{*})\\
&\quad {-} \frac{2}{n}(\widehat{\boldsymbol{\theta}}-\boldsymbol{\theta}^{*})^{\top}\mathbf{X}^{\top}(\mathbf{x}-\mathbf{X}\boldsymbol{\theta}^{*}).
\label{eq:full1}
\end{align}
To obtain bounds on the left hand side of \prettyref{eq:full1}, we bound the terms on the right hand side individually. The bound on the first term follows from \prettyref{eq:full prediction error} on the consistency of the LASSO under the RE and deviation bound assumptions (See \prettyref{prop:prediction error}),
\begin{align*}
\frac{1}{n}(\widehat{\boldsymbol{\theta}}-\boldsymbol{\theta}^{*})^{\top}\mathbf{X}^{\top}\mathbf{X}(\widehat{\boldsymbol{\theta}}-\boldsymbol{\theta}^{*}) \leq \frac{18m}{m+1}\frac{k\lambda_{n}^{2}}{\alpha},
\end{align*}
while the second terms can be bounded using \prettyref{eq:full l1 norm} as
\begin{align*}
\notag \left\vert\frac{1}{n}(\widehat{\boldsymbol{\theta}}\!-\!\boldsymbol{\theta}^{*})^{\top}\mathbf{X}^{\top}\!(\mathbf{x}\!-\!\mathbf{X}\boldsymbol{\theta}^{*})\right\vert &\! \leq \left\Vert \widehat{\boldsymbol{\theta}}\!-\!\boldsymbol{\theta}^{*}\right\Vert_{1} \left\Vert\frac{1}{n}\mathbf{X}^{\top}(\mathbf{x}\!-\!\mathbf{X}\boldsymbol{\theta}^{*})\right\Vert_{\infty}\\
\nonumber &\leq \frac{3m}{m+1}\frac{k\lambda_{n}^{2}}{\alpha},
\end{align*}
where we have used the choice of \({\lambda_n}\) and \prettyref{eq:DB-1} to conclude:
\begin{align*}
\left\Vert\mathbf{X}^{\top}(\mathbf{x}-\mathbf{X}\boldsymbol{\theta}^{*}) \right\Vert_{\infty} \leq \frac{\lambda_n}{4}. 
\end{align*}
Combining these two bounds via the triangle inequality concludes the proof.
\end{proof}

\begin{prop}[Deviation Result for the \emph{Reduced} Model] \label{prop:reduced model deviation}
Suppose $\widehat{\boldsymbol{\Sigma}}\sim\text{RE}\left(\alpha,\tau\right)$,
with $\tau$ satisfying \({32k\tau}/{\alpha}={(m-1)}/{m}\)
for some $m>1$ and \((\mathbf{X}_{(1)}, \mathbf{x})\) satisfying the deviation bound \eqref{eq:DB-2}. Let $J$ denote the support of $\boldsymbol{\theta}^*_{(1)}$, with its complement denoted by $J^c$. Then, for any 
\(
\lambda_{n}\geq4 \mathbbm{Q}'(\boldsymbol{\theta}^{*},\boldsymbol{\Sigma}_{\epsilon})\sqrt{{\log(2p)}/{n}},
\)
the solution to \emph{reduced} model in \prettyref{eq:Lasso-full&reducedmoedel} satisfies:
\begin{align}
\notag \left\vert \ell\Big(\widehat{\widetilde{\boldsymbol{\theta}}}_{(1)},\boldsymbol{0}\Big)-\ell\left(\widetilde{\boldsymbol{\theta}}^*_{(1)},\boldsymbol{0}\right)\right\vert & \leq 20\frac{k\lambda_{n}^{2}}{\alpha/m}\\
\notag & \quad +\left(8\sqrt{2m}+18\right)\lambda_{n}\left\Vert \widetilde{\boldsymbol{\theta}}^*_{(1)J^{c}}\right\Vert _{1}\\
&=:\Delta_{R},
\label{eq:full3}
\end{align}
\end{prop}
\begin{proof}
In a similar fashion to \prettyref{prop:full model deviation}, by invoking the consistency results of LASSO for the \emph{reduced} model under the RE and deviation bound assumptions (See \prettyref{prop:prediction error-reducedmodel}) and noting \prettyref{eq:DB-2} we have:
\begin{align}
\notag \!\left\vert \ell\Big(\widehat{\widetilde{\boldsymbol{\theta}}}_{(1)},\boldsymbol{0}\Big)\!-\!\ell\left(\widetilde{\boldsymbol{\theta}}^*_{(1)},\boldsymbol{0}\right)\!\right\vert\! \leq & \,
12\frac{k\lambda_{n}^{2}}{\alpha/m}\\
\notag & +\!\left(\!8\!\left(\!\sqrt{2m}\!+\!1\right)\!+\!2\!\right)\!\lambda_{n}\!\left\Vert \widetilde{\boldsymbol{\theta}}^*_{(1)J^{c}}\right\Vert _{1}\\
& + 16\sqrt{\frac{\lambda_{n}^{3}k}{\alpha/m}}\sqrt{\left\Vert \widetilde{\boldsymbol{\theta}}^*_{(1)J^{c}}\right\Vert_1}.
\label{eq:full4}
\end{align}
Simplifying the last term using Arithmetic Mean-Geometric Mean inequality proves the claim.
\end{proof}

To establish the consistency results used in \prettyref{prop:full model deviation} and \prettyref{prop:reduced model deviation}, we first state a result adapted from \cite{basu2015regularized} on the prediction error of LASSO under the \emph{full} model:

\begin{prop}[Prediction Error for the \emph{Full} Model]
\label{prop:prediction error} Suppose $\widehat{\boldsymbol{\Sigma}}\sim\text{RE}\left(\alpha,\tau\right)$, with $\tau$ satisfying \({32k\tau}/{\alpha}={(m-1)}/{m}\)
for some $m>1$ and \((\mathbf{X}, \mathbf{x})\) satisfying the deviation bound \eqref{eq:DB-1}. 
Then for any \(\lambda_{n}\geq4 \mathbbm{Q}(\boldsymbol{\theta}^{*},\boldsymbol{\Sigma}_{\epsilon})\sqrt{\frac{\log(2p)}{n}},\) the solution to the \emph{full} model in \prettyref{eq:Lasso-full&reducedmoedel} satisfies: 
\begin{align}
\Vert\widehat{\boldsymbol{\theta}}-\boldsymbol{\theta}^{*}\Vert_2 & \leq\frac{3m}{m+1}\frac{\sqrt{k}\lambda_{n}}{\alpha},\\
\Vert\widehat{\boldsymbol{\theta}}-\boldsymbol{\theta}^{*}\Vert_{1} & \leq\frac{12m}{m+1}\frac{k\lambda_{n}}{\alpha},\label{eq:full l1 norm}\\
(\widehat{\boldsymbol{\theta}}-\boldsymbol{\theta}^{*})^{\top}\widehat{\boldsymbol{\Sigma}}(\widehat{\boldsymbol{\theta}}-\boldsymbol{\theta}^{*}) & \leq\frac{18m}{m+1}\frac{k\lambda_{n}^{2}}{\alpha}\label{eq:full prediction error}.
\end{align}
\end{prop}
\begin{proof}
The proof closely follows that of \cite[Proposition 4.1]{basu2015regularized}, and is thus omitted for brevity.
\end{proof}

In what follows, we show that the particular choice of $\tau$ satisfying \({32k\tau}/{\alpha}={(m-1)}/{m}\) for some $m>1$ will simplify the prediction error analysis of the \emph{reduced} model. As for the \emph{reduced} model, the main technical difficulty in establishing prediction error bounds stems from the fact that $\widetilde{\boldsymbol{\theta}}^*_{(1)}$ is no longer $k$--sparse. We address this issue in the following proposition:

\begin{prop}[Prediction Error for the \emph{Reduced} Model]
\label{prop:prediction error-reducedmodel} Suppose $\widehat{\boldsymbol{\Sigma}}\sim\text{RE}\left(\alpha,\tau\right)$,
with $\tau$ satisfying the relation in \prettyref{prop:prediction error} and \((\mathbf{X}_{(1)}, \mathbf{x})\) satisfying the deviation bound \eqref{eq:DB-2}. Let $J$ denote the support of $\boldsymbol{\theta}^*_{(1)}$, with its complement denoted by $J^c$. Then, for any 
\(
\lambda_{n}\geq4 \mathbbm{Q}'(\boldsymbol{\theta}^{*},\boldsymbol{\Sigma}_{\epsilon})\sqrt{{\log(2p)}/{n}},
\)
the solution to \emph{reduced} model in \prettyref{eq:Lasso-full&reducedmoedel} satisfies the inequalities (\ref{eq:thm-6-1}), (\ref{eq:thm-6-2}) and (\ref{eq:thm-6-3}).
\begin{floateq}[b]{70}{0}{1}
{
\left\Vert \boldsymbol{\widehat{\widetilde{\theta}}}_{(1)}-\widetilde{\boldsymbol{\theta}}_{(1)}^{*}\right\Vert_2&\leq\frac{3}{2}\frac{\lambda_{n}\sqrt{k}}{\alpha/m}+\sqrt{\frac{2m}{k}}\left\Vert \widetilde{\boldsymbol{\theta}}^*_{(1)J^{c}}\right\Vert _{1}+\sqrt{\frac{4\lambda_{n}}{\alpha/m}\left\Vert \widetilde{\boldsymbol{\theta}}^*_{(1)J^{c}}\right\Vert _{1}}, \label{eq:thm-6-1}\\
\left \Vert\boldsymbol{\widehat{\widetilde{\theta}}}_{(1)}-\widetilde{\boldsymbol{\theta}}_{(1)}^{*}\right\Vert_{1}&\leq6\frac{\lambda_{n}k}{\alpha/m}+4(\sqrt{2m}+1)\left\Vert \widetilde{\boldsymbol{\theta}}^*_{(1)J^{c}}\right\Vert_1 +8\sqrt{\frac{\lambda_{n}k}{\alpha/m}}\sqrt{\left\Vert \widetilde{\boldsymbol{\theta}}^*_{(1)J^{c}}\right\Vert _{1}},\label{eq:thm-6-2}\\
\frac{1}{n}\left(\boldsymbol{\widehat{\widetilde{\theta}}}_{(1)}-\widetilde{\boldsymbol{\theta}}_{(1)}^{*}\right)^\top \mathbf{X}_{(1)}^{\top}\mathbf{X}_{(1)}\left(\boldsymbol{\widehat{\widetilde{\theta}}}_{(1)}-\widetilde{\boldsymbol{\theta}}_{(1)}^{*}\right) &\leq 9\frac{\lambda_{n}^{2}k}{\alpha/m} {+}\left(6\left(\sqrt{2m}+1\right)+2\right)\lambda_{n}\left\Vert \widetilde{\boldsymbol{\theta}}^*_{(1)J^{c}}\right\Vert_1
+12\sqrt{\frac{\lambda_{n}^{3}k}{\alpha/m}}\sqrt{\left\Vert \widetilde{\boldsymbol{\theta}}^*_{(1)J^{c}}\right\Vert_1 \label{eq:thm-6-3}}.}
\end{floateq}
\setcounter{equation}{\numexpr\value{equation}+3}
\end{prop}
\begin{proof}
Recall that $\widetilde{\boldsymbol{\theta}}^*_{(1)}:=\boldsymbol{\theta}_{(1)}^{*}+\mathbf{C}_{11}^{-1}\mathbf{C}_{21}{\boldsymbol{\theta}^*_{(2)}}$.
Let us define
\begin{align*}
F(\boldsymbol{\theta}_{(1)}):=\frac{1}{n}\left\Vert \mathbf{x}-\mathbf{X}_{(1)}\boldsymbol{\theta}_{(1)}\vphantom{-\mathbf{X}_{(2)}\boldsymbol{\theta}^{0}}\right\Vert_2^{2}+\lambda_{n}\Vert\boldsymbol{\theta}_{(1)}\Vert_{1}.
\end{align*}
and consider the quantity, $\Delta F\Big(\widehat{\widetilde{\boldsymbol{\theta}}}_{(1)}\Big):= F\Big(\boldsymbol{\widehat{\widetilde{\theta}}}_{(1)}\Big)-F\Big(\widetilde{\boldsymbol{\theta}}_{(1)}^*\Big)$.
Also, let $\mathbf{v}:=\boldsymbol{\theta}^*_{(1)}-\boldsymbol{\widehat{\widetilde{\theta}}}_{(1)}$ and $\mathbf{\widetilde{\mathbf{v}}}:=\mathbf{v}+\mathbf{C}_{11}^{-1}\mathbf{C}_{12}{\boldsymbol{\theta}^*_{(2)}}$. Then, $\Delta F\left(\widehat{\widetilde{\boldsymbol{\theta}}}_{(1)}\right)$
can be simplified as: 
\begin{align*} \Delta F\Big(\widehat{\widetilde{\boldsymbol{\theta}}}_{(1)}\Big) &= F\Big(\boldsymbol{\widehat{\widetilde{\theta}}}_{(1)}\Big)-F\Big(\widetilde{\boldsymbol{\theta}}_{(1)}^*\Big)\\
 &=\frac{1}{n}\widetilde{\mathbf{v}}^{\top}\mathbf{X}_{(1)}^{\top}\mathbf{X}_{(1)}\widetilde{\mathbf{v}}+
 \frac{2}{n}\mathbf{X}_{(1)}^{\top}\left(\mathbf{\!x}-\!\mathbf{X}_{(1)}\widetilde{\boldsymbol{\theta}}_{(1)}^{*}\vphantom{-\mathbf{X}_{(2)}\boldsymbol{\theta}^{0}}\right)\\
 & \quad + \lambda_{n}\left(\Big\Vert \widehat{\widetilde{\boldsymbol{\theta}}}_{(1)}\Big\Vert _{1}-\Big\Vert \widetilde{\boldsymbol{\theta}}_{(1)}^{*}\Big\Vert_{1}\right).
\end{align*}
Using the fact that
\begin{align*}
\left\Vert \frac{2}{n}\mathbf{X}_{(1)}^{\top}\left(\mathbf{\!x}-\!\mathbf{X}_{(1)}\widetilde{\boldsymbol{\theta}}_{(1)}^{*}\vphantom{-\mathbf{X}_{(2)}\boldsymbol{\theta}^{0}}\right)\right\Vert _{\infty} {\leq}  \frac{\lambda_{n}}{2},
\end{align*}
we get:
\begin{align*}
\Delta F\Big(\widehat{\widetilde{\boldsymbol{\theta}}}_{(1)}\Big)\geq& \frac{1}{n}\widetilde{\mathbf{v}}^{\top}\mathbf{X}_{(1)}^{\top}\mathbf{X}_{(1)}\widetilde{\mathbf{v}}-\frac{\lambda_{n}}{2}(\Vert\widetilde{\mathbf{v}}_{J}\Vert_{1}+\Vert\widetilde{\mathbf{v}}_{J^{c}}\Vert_{1})\\
& +\lambda_{n}\left(\Vert\widetilde{\mathbf{v}}_{J^{c}}\Vert_{1}-\Vert\widetilde{\mathbf{v}}_{J}\Vert_{1}-2\Big\Vert\widetilde{\boldsymbol{\theta}}^*_{(1)J^{c}}\Big\Vert_{1}\right),
\end{align*}
where we split the error $\widetilde{\mathbf{v}}$ into components within $J$ and components within $J^c$.
Now, since $\widetilde{\mathbf{v}}^{\top}\mathbf{X}_{(1)}^{\top}\mathbf{X}_{(1)}\widetilde{\mathbf{v}}$ is non-negative, we arrive at: 
\begin{align}
\notag 0 & \geq-\frac{\lambda_{n}}{2}\left(\left\Vert \widetilde{\mathbf{v}}_{J}\right\Vert _{1}+\left\Vert \widetilde{\mathbf{v}}_{J^{c}}\right\Vert _{1}\right)\\
\notag \quad & \quad +\lambda_{n}\left(\left\Vert \widetilde{\mathbf{v}}_{J^{c}}\right\Vert _{1}-\left\Vert \widetilde{\mathbf{v}}_{J}\right\Vert _{1}-2\left\Vert \widetilde{\boldsymbol{\theta}}^*_{(1)J^{c}}\right\Vert _{1}\right) \\
 & =-\frac{\lambda_{n}}{2}\left(3\Vert\widetilde{\mathbf{v}}_{J}\Vert_{1}-\Vert\widetilde{\mathbf{v}}_{J^{c}}\Vert_{1}+4\Big\Vert\widetilde{\boldsymbol{\theta}}^*_{(1)J^{c}}\Big\Vert_{1}\right).
 \label{eq:compress}
\end{align}
The rest of the treatment follows the derivation of weakly sparse or compressible models (See, for example, the derivation of the main theorem in \cite{negahban2012}). Using the inequality \prettyref{eq:compress}, we can write: 
\begin{align}
\notag \Vert\widetilde{\mathbf{v}}\Vert_{1}&=\Vert\widetilde{\mathbf{v}}_{J}\Vert_{1}+\Vert\widetilde{\mathbf{v}}_{J^{c}}\Vert_{1} \leq 4\Vert\widetilde{\mathbf{v}}_{J}\Vert_{1}+4\Big\Vert\widetilde{\boldsymbol{\theta}}^*_{(1)J^{c}}\Big\Vert_{1}\\
& \leq 4\sqrt{k}\Vert\widetilde{\mathbf{v}}_{J}\Vert_2+4\Big\Vert\widetilde{\boldsymbol{\theta}}^*_{(1)J^{c}}\Big\Vert_{1},\label{eq:vase-condition}
\end{align}
where the last inequality follows from the fact $\left|J\right|\leq k$. This inequality together with the RE condition and $\tau$ satisfying the relation in \prettyref{prop:prediction error}, allows to write:
\begin{align}
\notag \frac{1}{n}\widetilde{\mathbf{v}}^{\top}\mathbf{X}_{(1)}^{\top}\mathbf{X}_{(1)}\widetilde{\mathbf{v}} 
&\geq\alpha\Vert\widetilde{\mathbf{v}}\Vert^{2}_2-\alpha\frac{m\!-\!1}{m}\Vert\widetilde{\mathbf{v}}\Vert^{2}_2-\frac{\alpha}{m}\frac{m\!-\!1}{k}\Big\Vert\widetilde{\boldsymbol{\theta}}^*_{(1)J^{c}}\Big\Vert_{1}^{2}\\
&\geq\frac{\alpha}{m}\Vert\widetilde{\mathbf{v}}\Vert^{2}_2-\frac{\alpha}{m}\frac{m-1}{k}\Big\Vert\widetilde{\boldsymbol{\theta}}^*_{(1)J^{c}}\Big\Vert_{1}^{2}, \label{eq:restricted-RE-condition}
\end{align}
using the inequality $\sqrt{2(a^{2}+b^{2})}\geq(a+b)$.

Combining \prettyref{eq:compress} and \prettyref{eq:restricted-RE-condition}, we then finally arrive at:
\begin{align}
\nonumber \Delta F\Big(\widehat{\widetilde{\boldsymbol{\theta}}}_{(1)}\Big)
&\geq\frac{\alpha}{m}\Vert\widetilde{\mathbf{v}}\Vert^{2}_2 
-\frac{\alpha}{m}\frac{m-1}{k}\Big\Vert\widetilde{\boldsymbol{\theta}}^*_{(1)J^{c}}\Big\Vert_{1}^{2}\\
\notag & \qquad -\frac{\lambda_{n}}{2}\left(3\Vert\widetilde{\mathbf{v}}_{J}\Vert_{1}-\Vert\widetilde{\mathbf{v}}_{J^{c}}\Vert_{1}+4\Big\Vert\widetilde{\boldsymbol{\theta}}^*_{(1)J^{c}}\Big\Vert_{1}\right)\\
& \geq\frac{\alpha}{m}\Vert\widetilde{\mathbf{v}}\Vert^{2}_2 -\frac{\alpha}{m}\frac{m-1}{k}\Big\Vert\widetilde{\boldsymbol{\theta}}^*_{(1)J^{c}}\Big\Vert_{1}^{2}\\
& \qquad -\frac{\lambda_{n}}{2}\left(3\sqrt{k}\Vert\widetilde{\mathbf{v}}\Vert_2+4\Big\Vert\widetilde{\boldsymbol{\theta}}^*_{(1)J^{c}}\Big\Vert_{1}\right),
 \label{eq:ineq3}
\end{align}
where the last step follows from the inequalities $\Vert\widetilde{\mathbf{v}}_{J}\Vert_{1}\leq\sqrt{k}\Vert\widetilde{\mathbf{v}}_{J}\Vert_2\leq\sqrt{k}\Vert\widetilde{\mathbf{v}}\Vert_2$ and $\Vert\widetilde{\mathbf{v}}_{J^{c}}\Vert_{1}\geq0$. 

An application of \prettyref{lem:useful-lemma} establishes that the right hand side of the inequality \prettyref{eq:ineq3} will be positive if: 
\begin{align}
\notag \Vert\widetilde{\mathbf{v}}\Vert^{2} &\geq\frac{9}{4}\frac{\lambda_{n}^{2}}{\alpha^{2}/m^{2}}k\\
& \quad +\frac{\lambda_{n}}{\alpha/m}\left(\frac{2}{\lambda_{n}}\frac{\alpha}{m}\frac{m-1}{k}\Big\Vert \widetilde{\boldsymbol{\theta}}^*_{(1)J^{c}}\Big\Vert _{1}^{2} +4\Big\Vert \widetilde{\boldsymbol{\theta}}^*_{(1)J^{c}}\Big\Vert _{1}\right).
\end{align}
From the latter inequality, the first claim of the proposition follows using the fact that $\left\Vert \mathbf{a}\right\Vert _{1}\geq\left\Vert \mathbf{a}\right\Vert _{2}$; the second claim follows from the first claim together with \prettyref{eq:vase-condition}; the last claim follows from the fact that:
\begin{equation}
\frac{1}{n}\widetilde{\mathbf{v}}^{\top}\mathbf{X}_{(1)}^{\top}\mathbf{X}_{(1)}\widetilde{\mathbf{v}}\leq\frac{\lambda_{n}}{2}\left(3\Vert\widetilde{\mathbf{v}}_{J}\Vert_{1}-\Vert\widetilde{\mathbf{v}}_{J^{c}}\Vert_{1}+4\Big\Vert\widetilde{\boldsymbol{\theta}}^*_{(1)J^{c}}\Big\Vert_{1}\right),
\end{equation}
which concludes the proof of the proposition.
\end{proof}

\section{Verification of Conditions \ref{cond:RE} and \ref{cond:DB}}\label{sec:verif}
Having established Propositions \ref{prop:full model deviation} and \ref{prop:reduced model deviation}, it only remains to show that \prettyref{cond:RE} and \prettyref{cond:DB} hold with high probability, under both the \emph{full} and \emph{reduced} models. 

The following proposition establishes that the RE condition (\prettyref{cond:RE}) holds with high probability:

\begin{prop}[Verifying RE for $\widehat{\boldsymbol{\Sigma}}=\mathbf{X}^{\top}\mathbf{X}/n)$]
Let \label{prop:RE condition}
\begin{align*}
&\zeta:=54\frac{\Lambda_{\max}(\boldsymbol{\Sigma}_{\epsilon})/\mu_{\min}(\breve{\mathbf{{A}}})}{\Lambda_{\min}(\boldsymbol{\Sigma}_{\epsilon})/\mu_{\max}(\mathbf{A})} ,\quad
\alpha:=\frac{\Lambda_{\min}(\boldsymbol{\Sigma}_{\epsilon})}{2\mu_{\max}(\mathbf{A})},\\
& \qquad \qquad  \tau:=\frac{4\alpha\max\{\zeta^{2},1\}}{c}  \frac{\log(2p)}{n}.
 \end{align*}
Then, for $n\ge C_0 \max\{\zeta^{2},1\}k\log(2p)$, there exist constants $c_{1},c_{2}$ such that
\begin{equation}
\mathbbm{P}\left[\widehat{\boldsymbol{\Sigma}}\sim\text{RE}\left(\alpha,\tau\right)\right]\geq1-c_{1}\exp(-c_{2}n\min\{\zeta^{-2},1\}).\label{eq:RE-cond}
\end{equation}
\end{prop}
\begin{proof}
The proof closely follows that of \cite[proof of Proposition 4.2]{basu2015regularized}, and is thus omitted for brevity.
\end{proof}

\begin{remark}[]
In order for $\tau$ to satisfy \({32k\tau}/{\alpha}={(m-1)}/{m}\) for some $m>1$, we precisely need \(n\geq (128/c)(m/(m-1))\max\{\zeta^{2},1\}k\log(2p)\).
\end{remark}

Finally, the following two propositions establish that the deviation conditions (\prettyref{cond:DB}) hold with high probability.

\begin{prop}[Deviation Condition for the \emph{full} Model]
For $n \ge D_0 \log(2p)$, there exist constants $d_{0},d_{1}$ and $d_2 > 0$
such that \label{prop:deviation condition}
\begin{equation}
\mathbbm P\!\left[\!\left\Vert \frac{1}{n}\mathbf{X}^{\top}\!({\mathbf{x\!}-\!\mathbf{X}\boldsymbol{\theta}^{*}})\right\Vert _{\infty}\!\!\!\!\geq \mathbbm Q(\boldsymbol{\theta}^{*},\Sigma_{\epsilon})\sqrt{\frac{\log(2p)}{n}}\right]\!\leq\!\frac{d_{1}}{(2p)^{d_2}},\label{eq:DB-1-proof}
\end{equation}
where 
\begin{align*}
\mathbbm{Q}(\boldsymbol{\theta}^{*},\boldsymbol{\Sigma}_{\epsilon}):=d_0 \Lambda_{\max}(\boldsymbol{\Sigma}_{\epsilon})\left(1+\frac{1+\mu_{\max}(\mathbf{A})}{\mu_{\min}(\mathbf{A})}\right).
\end{align*}
\end{prop}
\begin{proof}
The proof follows that of \cite[Proposition 4.3]{basu2015regularized}, and is thus omitted for brevity.
\end{proof}

\begin{prop}[Deviation Condition for the \emph{Reduced} Model]
For $n \ge D_0' \log(2p)$, there exist constants $d'_{0},d'_{1}$, and $d'_2 > 0$
such that \label{prop:deviation condition-2}
\begin{alignat}{1}
\mathbbm{P}\!\left[\!\left\Vert \frac{1}{n}\mathbf{X}_{(1)}^{\top}\! \!\left(\mathbf{x}\!-\!\mathbf{X}_{(1)}\widetilde{\boldsymbol{\theta}}^*_{(1)}\!\right)\!\right\Vert _{\infty}\!\!\!\!\!\geq \! \mathbbm Q'(\boldsymbol{\theta}^{*},\boldsymbol{\Sigma}_{\epsilon})\!\sqrt{\frac{\log(2p)}{n}}\right] \!\!\leq \!\frac{d'_{1}}{{(2p)}^{d_2'}},\label{eq:DB-2-proof}
\end{alignat}
where \({\mathbbm Q'(\boldsymbol{\theta}^{*},\boldsymbol{\Sigma}_{\epsilon})}\) is defined as in (\ref{eq:Q-dash-def}).
\begin{floateq}[b]{83}{0}{1}
{
&\hspace*{7em}\mathbbm{Q}'(\boldsymbol{\theta}^{*},\Sigma_{\epsilon}):= d'_0\Lambda_{\max}(\boldsymbol{\Sigma}_{\epsilon})\left(\vphantom{\frac{3\left\Vert\left[-\mathbf{C}_{11}^{-1}\mathbf{C}_{12}{\boldsymbol{\theta}^*_{(2)}};{\boldsymbol{\theta}^*_{(2)}}\right]\right\Vert_2}{\mu_{\min}(\breve{\mathbf{A}})}} 1+\frac{1+\mu_{\max}(\mathbf{A})}{\mu_{\min}(\mathbf{A})}\right.\left. +\frac{3\left\Vert\left[-\mathbf{C}_{11}^{-1}\mathbf{C}_{12}{\boldsymbol{\theta}^*_{(2)}};{\boldsymbol{\theta}^*_{(2)}}\right]\right\Vert_2}{\mu_{\min}(\breve{\mathbf{A}})}\right). \label{eq:Q-dash-def} \\
&\mathbbm P\left[\left\vert \mathbf{e}_{i}^{\top}\frac{1}{n}\mathbf{X}^{\top}\mathbf{X}\left[-\mathbf{C}_{11}^{-1}\mathbf{C}_{12}{\boldsymbol{\theta}^*_{(2)}};{\boldsymbol{\theta}^*_{(2)}}\right]\right\vert \geq3\frac{\Lambda_{\max}(\boldsymbol{\Sigma}_{\epsilon})}{\mu_{\min}(\breve{\mathbf{A}})}\eta \left\Vert \left[-\mathbf{C}_{11}^{-1}\mathbf{C}_{12}{\boldsymbol{\theta}^*_{(2)}};{\boldsymbol{\theta}^*_{(2)}}\right]\right\Vert_2 
\vphantom{\frac{\Lambda_{\max}(\boldsymbol{\Sigma}_{\epsilon})}{\mu_{\min}(\breve{\mathbf{A}})}} \right] 
 \leq 6\exp[-cn\min\{\eta,\eta^{2}\}]. \label{eq:prop9-1}\\
 & \mathbbm{P}\left[ \left\vert \frac{1}{n}\mathbf{e}_{i}^{\top}\left(\mathbf{X}^{\top}\mathbf{X}\left[-\mathbf{C}_{11}^{-1}\mathbf{C}_{12}{\boldsymbol{\theta}^*_{(2)}};{\boldsymbol{\theta}^*_{(2)}}\right]+\mathbf{X}^{\top}\boldsymbol{\epsilon}\right)\right\vert \geq \Lambda_{\max}(\boldsymbol{\Sigma}_{\epsilon})
\left(1+\frac{1+\mu_{\max}(\mathbf{A})}{\mu_{\min}(\mathbf{A})} +\vphantom{+\frac{3\left\Vert\left[-\mathbf{C}_{11}^{-1}\mathbf{C}_{12}{\boldsymbol{\theta}^*_{(2)}};{\boldsymbol{\theta}^*_{(2)}}\right]\right\Vert_2}{\mu_{\min}(\breve{\mathbf{A}})}}\frac{3\left\Vert\left[-\mathbf{C}_{11}^{-1}\mathbf{C}_{12}{\boldsymbol{\theta}^*_{(2)}};{\boldsymbol{\theta}^*_{(2)}}\right]\right\Vert_2}{\mu_{\min}(\breve{\mathbf{A}})}\right)\eta
\right] \nonumber \\
&\hspace*{35em}\leq12\exp[-cn\min\{\eta,\eta^{2}\}]. \label{eq:prop9-2}
}
\end{floateq}
\end{prop}
\begin{proof}
In the \emph{reduced} model, the deviation can be expressed as:
\begin{align}
\hphantom{=}&\nonumber \frac{1}{n}\mathbf{X}_{(1)}^{{\top}}\left(\mathbf{x}-\mathbf{X}_{(1)}\widetilde{\boldsymbol{\theta}}^*_{(1)}-\mathbf{X}_{(2)}\boldsymbol{0}\right) \\
\nonumber =&\frac{1}{n}\mathbf{X}_{(1)}^{\top}\left(-\mathbf{X}_{(1)}\mathbf{C}_{11}^{-1}\mathbf{C}_{12}{\boldsymbol{\theta}^*_{(2)}}+\mathbf{X}_{(2)}{\boldsymbol{\theta}^*_{(2)}}+\boldsymbol{\epsilon}\right)\\
\nonumber =&-\frac{1}{n}\mathbf{X}_{(1)}^{\top}\mathbf{X}_{(1)}\mathbf{C}_{11}^{-1}\mathbf{C}_{12}{\boldsymbol{\theta}^*_{(2)}}+\frac{1}{n}\mathbf{X}_{(1)}^{\top}\mathbf{X}_{(2)}{\boldsymbol{\theta}^*_{(2)}}+\frac{1}{n}\mathbf{X}_{(1)}^{\top}\boldsymbol{\epsilon}.
\end{align}
The last term can be bounded in a similar fashion as done for the deviation in \prettyref{prop:deviation condition}. The $i^{\sf th}$ component of the first two terms can be expressed as $\mathbf{e}_{i}^{\top}\frac{1}{n}\mathbf{X}^{\top}\mathbf{X}\left[-\mathbf{C}_{11}^{-1}\mathbf{C}_{12}{\boldsymbol{\theta}^*_{(2)}};{\boldsymbol{\theta}^*_{(2)}}\right]$,
for $i=1,2,\cdots,p$ where $\mathbf{e}_{i}$'s are the standard unit bases in $\mathbb{R}^{2p}$. Invoking \cite[Proposition 2.4(a)]{basu2015regularized} and noting that $\mathcal{M}(\mathbf{F})\leq{\Lambda_{\max}(\boldsymbol{\Sigma}_{\epsilon})}\Big/{\mu_{\min}(\breve{\mathbf{A}})}$, we get (\ref{eq:prop9-1}). Using the union bound, we can then arrive at (\ref{eq:prop9-2}). Using the latter inequality, the statement of the proposition follows from the same arguments used in the proof of \cite[Proposition 4.3]{basu2015regularized}. 
\setcounter{equation}{\numexpr\value{equation}+3}
\end{proof}
\section{Concentration Inequalities and Technical Lemmas\label{sec:Concentration-Lemmas}}

\setcounter{thm}{0}
\begin{lem}\label{lem:normal_conc}
Given i.i.d. samples from a normal distribution, i.e., $w_{t}\sim\mathcal{N}(0,\sigma^{2})$,
the following holds:
\begin{equation}
\mathbbm P\left[\left\vert \frac{1}{n}\sum_{t=1}^{n}\frac{w_{t}^{2}}{\sigma^{2}}-1\right\vert \geq t\right]\leq2\exp\left(-\frac{nt^{2}}{8}\right), \, \forall t \in (0, 1).
\end{equation}
\end{lem}
\begin{proof}
Define $z_{t}={w_{t}}/{\sigma}\sim\mathcal{N}(0,1)$, then $\sum_{t=1}^{n}z_{t}^{2}\sim\chi^{2}(n)$.
Clearly, $z_{t}^{2}$ is sub-exponential with parameters $(2,4)$, so is the
sum $\sum_{t=1}^{n}z_{t}^{2}$ with parameters $(2\sqrt{n},4)$. The claim of the 
lemma then follows from standard sub-exponential tail bounds.
\end{proof}
\begin{lem}[Concentration of the \emph{Full} Model Deviation]\label{lem:fullmodel_error}
Under the \emph{full} model, we have{, for $\Delta_{N} < \left(\boldsymbol{\Sigma}_\epsilon\right)_{1,1}$}: 
\begin{align*}
\mathbbm{P}\left[\Big|\ell\Big(\boldsymbol{\theta}^*_{(1)}, \boldsymbol{\theta}^*_{(2)}\Big) - (\boldsymbol{\Sigma}_\epsilon)_{1,1}\Big|\geq \Delta_N\right]\leq2\exp\left(-\frac{n\Delta_{N}^{2}}{8(\boldsymbol{\Sigma}_\epsilon)_{1,1}}\right).
\end{align*}
\end{lem}
\begin{proof}
Note that $\ell\Big(\boldsymbol{\theta}^*_{(1)}, \boldsymbol{\theta}^*_{(2)}\Big)=\sum_{t=1}^{n}\epsilon_{t}^{2}/n$.
Since $\epsilon_{t}\sim\mathcal{N}(0,(\boldsymbol{\Sigma}_\epsilon)_{1,1})$, using \prettyref{lem:normal_conc} we get:
\begin{align}
\mathbbm{P}\left[\left\vert \frac{1}{n}\sum_{t=1}^{n}\frac{\epsilon_{t}^{2}}{(\boldsymbol{\Sigma}_\epsilon)_{1,1}}-1\right\vert \geq t\right]\leq2\exp\left(-\frac{nt^{2}}{8}\right), \, \forall t \in (0, 1).
\end{align}
By letting $\Delta_N := t (\boldsymbol{\Sigma}_\epsilon)_{1,1}$, the claim of the lemma follows.
\end{proof}
\begin{lem}[Concentration of the \emph{Reduced} Model Deviation]
Suppose that the deviation conditions (\prettyref{cond:DB}) hold. Then, there exist constants $c_{3}$ and $c_{4}> 0$ such that 
\begin{align}
\left\vert\ell\Big(\widetilde{\boldsymbol{\theta}}^*_{(1)},\boldsymbol{0}\Big)-\ell\Big(\boldsymbol{\theta}^{*}_{(1)}, \boldsymbol{\theta}^{*}_{(2)}\Big)-D\right\vert \leq \Delta_{D},
\end{align}
with probability at least \begin{align*}1-c_{3}\exp\left(-c_{4}n\min\left\{\zeta^{-2}\frac{\log(2p)}{n},1\right\}\right),\end{align*} 
\begin{align*}
\notag \Delta_D :=& \, \, \mathbbm Q(\boldsymbol{\theta}^{*},\boldsymbol{\Sigma}_\epsilon)\sqrt{\frac{\log(2p)}{n}} \left\Vert\left[-\mathbf{C}_{11}^{-1}\mathbf{C}_{12}{\boldsymbol{\theta}^*_{(2)}};{\boldsymbol{\theta}^*_{(2)}}\right]\right\Vert_1\\
& \, + \frac{\alpha}{27}\left\Vert\left[-\mathbf{C}_{11}^{-1}\mathbf{C}_{12}{\boldsymbol{\theta}^*_{(2)}};{\boldsymbol{\theta}^*_{(2)}}\right]\right\Vert_2^{2},
\end{align*}
and
\begin{align*}
D:={\boldsymbol{\theta}^*_{(2)}}^{\top}(\mathbf{C}_{22}-\mathbf{C}_{21}\mathbf{C}_{11}^{-1}\mathbf{C}_{12}){\boldsymbol{\theta}^*_{(2)}}.
\end{align*}
with \({\alpha}\) and \(\mathbbm{Q}(\boldsymbol{\theta}^{*},\boldsymbol{\Sigma}_\epsilon)\) defined in \prettyref{prop:RE condition} and \prettyref{prop:deviation condition},  respectively.
\label{lem:reducedmodel-fullmodel_error}
\end{lem}
\begin{proof}
Since $\widetilde{\boldsymbol{\theta}}^*_{(1)}=\boldsymbol{\theta}^*_{(1)}+\mathbf{C}_{11}^{-1}\mathbf{C}_{12}\boldsymbol{\theta}^*_{(2)}$, we have:
\begin{align*}
\ell\Big(\widetilde{\boldsymbol{\theta}}^*_{(1)},\boldsymbol{0}\Big)&-\ell\Big(\boldsymbol{\theta}^{*}_{(1)}, \boldsymbol{\theta}^{*}_{(2)}\Big)\\
&=\frac{1}{n}\left\Vert \mathbf{x}-\mathbf{X}_{(1)}\widetilde{\boldsymbol{\theta}}^*_{(1)}-\mathbf{X}_{(2)}\boldsymbol{0}\right\Vert_2^{2}\\
& \quad -\frac{1}{n}\left\Vert \mathbf{x}-\mathbf{X}\boldsymbol{\theta}^{*}\right\Vert_2^{2}\\
 & =\frac{1}{n}\left\Vert-\mathbf{X}_{(1)}\mathbf{C}_{11}^{-1}\mathbf{C}_{12}{\boldsymbol{\theta}^*_{(2)}}+\mathbf{X}_{(2)}{\boldsymbol{\theta}^*_{(2)}}+\boldsymbol{\epsilon}\right\Vert_2^{2}\\
 & \quad -\frac{1}{n}\Vert\boldsymbol{\epsilon}\Vert_2^{2}\\
 & =\frac{2}{n}\boldsymbol{\epsilon}^{\top}\mathbf{X}\boldsymbol{\vartheta}+\boldsymbol{\vartheta}^{\top}\frac{1}{n}\mathbf{X}^{\top}\mathbf{X}\boldsymbol{\vartheta},
\end{align*}
where \( \boldsymbol{\vartheta}:=\begin{bmatrix}(-\mathbf{C}_{11}^{-1}\mathbf{C}_{12}{\boldsymbol{\theta}^*_{(2)}})^{\top}, {\boldsymbol{\theta}^{*\top}_{(2)}}\end{bmatrix}^{\top}\).
Using the \prettyref{cond:DB}, we get: 
\begin{align*}
\begin{alignedat}{1}\left|\frac{1}{n}\boldsymbol{\epsilon}^{\top}\mathbf{X}\boldsymbol{\vartheta}\right|\! & \leq\left\Vert \frac{1}{n}\mathbf{X}^{\top}\epsilon\right\Vert _{\infty}\Vert\boldsymbol{\vartheta}\Vert_{1}\\
 & \leq \mathbbm Q(\boldsymbol{\theta}^{*},\boldsymbol{\Sigma}_\epsilon)\sqrt{\frac{\log(2p)}{n}}\Vert\boldsymbol{\vartheta}\Vert_{1}\\
 & \leq \mathbbm Q(\boldsymbol{\theta}^{*},\boldsymbol{\Sigma}_\epsilon)\sqrt{\frac{\log(2p)}{n}}\left(\left\Vert \mathbf{C}_{11}^{-1}\mathbf{C}_{12}\right\Vert_{1}\!+\!1\right)\Big\Vert{\boldsymbol{\theta}^*_{(2)}}\Big\Vert_{1}.
\end{alignedat}
\end{align*}
Furthermore, from \cite[Proposition 2.4]{basu2015regularized}, for any $\boldsymbol{\vartheta}\in\mathbbm{R}^{2p}$ and $\eta\ge0$, there exists a constant $c>0$ such that:
\begin{alignat*}{1}
&\mathbbm{P}\left[\left\vert \boldsymbol{\vartheta}^{\top}\!\left(\!\frac{1}{n}\mathbf{X}^{\top}\!\mathbf{X}\!-\!\mathbf{C}\right)\!\boldsymbol{\vartheta}\right\vert \geq\eta\Vert\boldsymbol{\vartheta}\Vert^{2}\frac{\Lambda_{\max}(\boldsymbol{\Sigma}_{\epsilon})}{\mu_{\min}(\breve{\mathbf{A}})}\!\right]\\
& \leq 2\exp[-cn\min\{\eta,\eta^{2}\}].
\end{alignat*}
Next, with the choice of 
\(\eta=\zeta^{-1}\sqrt{\log{2p} /n},\)
the latter concentration inequality establishes that:
\begin{align*}
\left\vert \boldsymbol{\vartheta}^{\top}\!\left(\frac{1}{n}\mathbf{X}^{\top}\mathbf{X}\right)\!\boldsymbol{\vartheta}-\boldsymbol{\vartheta}^{\top}\mathbf{C}\boldsymbol{\vartheta}\right\vert \leq
\frac{\alpha}{27}\sqrt{\frac{\log(2p)}{n}}\left\Vert \boldsymbol{\vartheta}\right\Vert_2^2,
\end{align*}
with probability at least $1-2\exp(-cn\min\{\zeta^{-2}\log(2p)/n,1\})$.
This along with the following observation concludes the proof of the lemma:
\(\boldsymbol{\vartheta}^{\top}\mathbf{C}\boldsymbol{\vartheta}={\boldsymbol{\theta}^{*\top}_{(2)}}(\mathbf{C}_{22}-\mathbf{C}_{21}\mathbf{C}_{11}^{-1}\mathbf{C}_{12}){\boldsymbol{\theta}^*_{(2)}}.\)
\end{proof}

\begin{lem}[Fano's Inequality]
\label{lem:ar_fano_ineq}
Let $\mathcal{Z}$ be a class of densities with a subclass $\mathcal{Z}^\star$ of densities $f_{\boldsymbol{\theta}_{(2)}^i}$, parameterized by $\boldsymbol{\theta}_{(2)}^{i}$, for $i \in \{0,\cdots,2^M \}$. Suppose that for any two distinct $\boldsymbol{\theta}_{(2)}^i, \boldsymbol{\theta}_{(2)}^j \in \mathcal{Z}^\star$, 
$\mathcal{D}_{\sf KL} \left(f_{\boldsymbol{\theta}_{(2)}^i}\Big\| f_{\boldsymbol{\theta}_{(2)}^j}\right) \leq \xi$ for some constant $\xi$. Let $\widehat{\boldsymbol{\theta}}$ be an estimate of the parameters. Then
\vspace{-.2cm}
\begin{equation}
\label{eq:fano}
\sup_j \mathbb{P}\left[\widehat{\boldsymbol{\theta}}_{(2)}\neq {\boldsymbol{\theta}_{(2)}^j}|H_j\right] \geq 1 - \frac{\xi+\log2}{M},
\end{equation}
where $H_j$ denotes the hypothesis that $\boldsymbol{\theta}_j$ is the true parameter, and induces the probability measure $\mathbb{P}[\cdot|H_j]$.
\end{lem}
\begin{proof}
See \cite[Page 323]{ibragimov2013statistical}. 
\end{proof}

Finally, the following elementary lemma provides a useful technical tool for simplifying some of the algebraic inequalities:

\begin{lem}
\label{lem:useful-lemma} The quadratic function $f(x)=ax^{2}-bx-c$ with $a,b,c>0$ is positive for all real $x$ satisfying 
\begin{align*}
x^{2}\geq\left(\frac{b}{a}\right)^{2}+2 \frac{c}{a}.
\end{align*}
\end{lem}
\begin{proof}
Clearly, $f(x)$ has only one positive root, denoted here by $x_{+}$, which satisfies ${f(x) > 0}, \text{ }\forall\text{ }x>x_{+}$. Using the following instance of Jensen's inequality, 
\begin{align*}
\sqrt{1+z}\leq\frac{1}{2}z+1 \hspace*{1em}\forall\hspace*{1em} z>0,
\end{align*} 
the positive root $x_+$ can be upper bounded as:
\begin{align*}
x_{+}^{2} & =\left(\frac{b}{2a}+\sqrt{\left(\frac{b}{2a}\right)^{2}+\frac{c}{a}}\right)^{2}\\
 & =2\left(\frac{b}{2a}\right)^{2}+\frac{c}{a}+2\frac{b}{2a}\sqrt{\left(\frac{b}{2a}\right)^{2}+\frac{c}{a}}\\
 & \leq2\left(\frac{b}{2a}\right)^{2}+\frac{c}{a}+2\left(\frac{b}{2a}\right)^{2}\left(\frac{1}{2}\left(\frac{2a}{b}\right)^{2}\frac{c}{a}+1\right)\\
 & =\left(\frac{b}{a}\right)^{2}+2\frac{c}{a} = x^2_0
\end{align*}
Then, $x^{2}\geq x_{0}^{2}$ implies $ax^{2}-bx-c>0$.
\end{proof}

\section*{Acknowledgment}
The authors would like to thank Patrick L. Purdon and Emery N. Brown for providing the data from \cite{lewis2012rapid}.


\end{document}